\documentclass[hidelinks,11pt]{article}
\usepackage{JASA_manu}
\usepackage{easyReview} 

\usepackage{natbib,color,hyperref,bm,amsmath,amsthm,amssymb,graphicx,float,booktabs,multirow,threeparttable,textcomp,times,accents,subfigure,ragged2e,makecell}
\usepackage{tikz}
\usetikzlibrary{positioning,shapes.geometric}

\theoremstyle{plain}
\newtheorem{theorem}{Theorem}

\newtheorem{assumption}{Assumption}
\newtheorem{proposition}{Proposition}
\newtheorem{example}{Example}

\makeatletter
\def\pr{ {f}}
\def\logit{\text{\rm logit }}
\def\expit{\text{\rm expit }}
\def\IF{{\rm  IF}}
\def\ipw{{\rm ipw}}
\def\aipw{{\rm aipw}}
\def\reg{{\rm reg}}
\def\dr{{\rm dr}}

\def\npr{{\rm np}}
\def\pa{{\rm par}}
\def\mar{{\rm mar }}

\def\M{{\mathcal M}}

\def\t{{ \mathrm{\scriptscriptstyle T} }}
\makeatother

\usepackage{xr}

\makeatletter
\newcommand*{\addFileDependency}[1]{
  \typeout{(#1)}
  \@addtofilelist{#1}
  \IfFileExists{#1}{}{\typeout{No file #1.}}
}
\makeatother

\newcommand*{\myexternaldocument}[1]{%
    \externaldocument{#1}%
    \addFileDependency{#1.tex}%
    \addFileDependency{#1.aux}%
}

\makeatletter
\newcommand*{\ind}{%
	\mathbin{%
		\mathpalette{\@ind}{}%
	}%
}
\newcommand*{\nind}{%
	\mathbin{
		\mathpalette{\@ind}{\not}
	}%
}
\newcommand*{\@ind}[2]{%
	\sbox0{$#1\perp\m@th$}
	\sbox2{$#1=$}
	\sbox4{$#1\vcenter{}$}
	\rlap{\copy0}
	\dimen@=\dimexpr\ht2-\ht4-.2pt\relax
	\kern\dimen@
	{#2}%
	\kern\dimen@
	\copy0 
} 
\makeatother

\myexternaldocument{supplement}

\graphicspath{{../../Rscripts/Dec082021/}}

\begin{document}
%


\title{\bf A stableness of resistance model for nonresponse adjustment with callback data}
\date{}
\author{
Wang Miao\thanks{Co-corresponding authors: Wang Miao,  Department of Probability and Statistics, Peking University, 
5 Yiheyuan Road, Beijing 100871, China;  E-mail: mwfy@pku.edu.cn. Baoluo Sun,  Department of Statistics and Data Science, National University of Singapore, Singapore 119077; E-mail: stasb@nus.edu.sg}, 
\, 
Xinyu Li, \,
Ping Zhang,\\
Department of Probability and Statistics, Peking University, Beijing, China\\
and 
Baoluo Sun$^*$ \\Department of Statistics and Data Science, National University of Singapore, Singapore
}

\maketitle


\begin{center}
\textbf{Summary}
\end{center}
Nonresponse   arises frequently   in surveys  and follow-ups   are  routinely made to  increase the response rate.
In order to monitor   the follow-up  process, callback data   have been  used in social sciences and survey studies   for decades. 
In modern surveys, the availability of callback data  is increasing   because the response rate is decreasing and  follow-ups  are essential to collect maximum information.
Although  callback data  are  helpful  to reduce the bias in surveys,  such data  have not been widely used  in statistical analysis until recently.
We propose a stableness of resistance assumption for nonresponse adjustment with callback data.
We establish the identification and the semiparametric efficiency theory under this assumption, and propose a suite of semiparametric estimation methods including doubly robust estimators, which generalize existing parametric approaches for  callback data analysis. 
We apply the approach to a Consumer Expenditure Survey dataset. 
The results  suggest an association between nonresponse and high housing expenditures.


\vspace*{.2in}




\section{Introduction}
Nonresponse often leads to substantially biased statistical inference and is   frequently encountered  in surveys and observational studies in many  areas of scientific research.
It has been a persistent concern of  statisticians and applied researchers  for many years.
The missingness is said to be  at random  (MAR)  or ignorable    if it does not depend on  the missing values conditional on fully-observed covariates, 
and otherwise it is called missing not at random (MNAR) or nonignorable.
A large body of work for nonresponse adjustment has been based on MAR.
However, there is  suspicion that the missingness mechanism is MNAR   in many  situations.
For example,  social stigma in sensitive  questions (e.g., HIV status, income, or drug use) makes nonresponse   dependent on unobserved variables. 
The missing data process  and the outcome distribution are identified under MAR, 
but under  MNAR  identification in general  fails to hold without extra information, which substantially jeopardizes statistical inference.
Identification means that the parameter or distribution of interest is uniquely determined from the observed-data distribution;
it is crucial for missing data analysis, without which statistical inference may be misleading and is of limited interest. 
Although one can achieve identification if the impact of the   outcome on the missingness is completely known,
this  should be used rather as a sensitivity analysis because in practice such information is seldom available, see e.g. \cite{robins2000sensitivity}.
Otherwise,  identification does not hold even for fully parametric models, except for several  fairly restrictive ones \citep[e.g.,][]{heckman1979sample,miao2016identifiability}.
For semiparametric and nonparametric models, identification and inference    under MNAR essentially require the use of auxiliary data, for example, 
 an instrumental variable    \citep{sun2018semiparametric,liu2020identification,tchetgen2017general} or    a shadow variable      \citep{d2010new,wang2014instrumental,miao2016varieties}.
Researchers have traditionally sought  such auxiliary variables   from the sampling frame, 
but it could be    difficult  in multipurpose studies where   multiple  survey variables are concerned and  multiple auxiliary variables are necessitated.
Moreover, such auxiliary variable methods  break down if  the auxiliary variables also have missing values  due to failure of contact in surveys.

In contrast to  the paucity of auxiliary variables  in     the sampling frame,  
callback data  offer an important  source of auxiliary information for nonresponse adjustment.
In the presence of nonresponse, the interviewer may  continue  to contact nonrespondents, 
and  the  contact process is recorded with callback data, sometimes called level-of-effort  data \citep{biemer2013using}.
For instance in the 2018 Consumer Expenditure Survey,  the maximum number of  contact attempts the interviewers made is about 30,
and  the number of calls made  to contact each unit  is recorded with the callback data.
Callback data have been widely used  in  epidemiological,  economic, social and political  surveys  for  a long time  since the 1940's.
Such  data are particularly useful    to monitor response rates and to study how  design features  affect  the data collection.
Examples  include \citet{politz1949attempt,filion1976exploring,drew1980modeling,lin1995using,potthoff1993correcting,wood2006using,jackson2010much,peress2010correcting,biemer2013using,clarsen2021revisiting}, and see \citet{groves1998nonresponse,olson2013paradata,kreuter2013improving} for a comprehensive review.
Although not all surveys   could provide callback data, their availability   is increasing in modern surveys.
The underlying  reason  is that the response rate in modern surveys is decreasing and  follow-ups are essentially required.
Examples include the National Health Interview Survey (NHIS),  the National Survey of Family Growth (NSFG), the National Survey on Drug Use and Health (NSDUH), 
the European Social Survey (ESS), the Behavioral Risk Factor Surveillance System (BRFSS), 
and the Consumer Expenditure Survey (CES).

However, callback data  have not been widely used  in statistical analysis until recently,
although  their usefulness  for nonresponse adjustment  has been recognized since \citet{politz1949attempt}.
The callback design is analogous to the two/multi-phase sampling \citep{deming1953probability} by viewing the  follow-ups as  the second-phase sample,   
but  it differs in that nonresponse  may still occur in the follow-up sample and is possibly not at random.
Hence, nonresponse adjustment remains difficult even if callbacks are available, due to the challenge for  identification and inference under MNAR.
In order to do nonresponse adjustment, the notion of  ``continuum of resistance''    assumes   that nonrespondents are more similar to delayed respondents than they are to early respondents,
so that the most reluctant  respondents are used to approximate the nonrespondents. 
This assumption has been   asserted in social and survey researches  for decades despite  conflicting evidence of its validity \citep{lin1995using,clarsen2021revisiting}.
Other approaches   include  modeling the joint likelihood of   callbacks and   frame variables \citep{biemer2013using};
Heckman-type models  \citep{chen2018generalization,zhang2018bayesian}; sensitivity analysis   \citep[Rotnitzky and Robins, 1995  unpublished manuscript;][]{daniels2015pattern}; etc.
Most notably, \cite{alho1990adjusting}, \cite{kim2014propensity}, \cite{qin2014semiparametric}, and \cite{guan2018semiparametric} use callback data to estimate response    propensity scores and propose  inverse probability weighted and empirical likelihood-based estimation methods,
where they impose  a fully parametric linear logistic propensity score model and     a common-slope assumption  that individual characteristics   influence the missingness process in the same way across the survey contacts.
These   previous approaches   rest on strong parametric models to achieve identification, which  circumvent the underlying sources for identification and limit their use in complex data application.

In order to  further investigate  the    usefulness of callback data and promote their application,
we take a   fundamentally nonparametric identification strategy, accompanied with practical parametric/semiparametric estimation methods.
Our contributions are threefold.
First, we characterize  a \textit{stableness of resistance} assumption and  establish the   identification    under this assumption, 
which is a nonparametric generalization of  the  parametric approach of \citet{alho1990adjusting}.
The stableness of resistance assumption    states that   the   impact  of the missing outcome  on  the  response propensity  is stable in the first two call attempts.
This assumption does not impose  parametric functional restrictions  on the propensity score, 
does not restrict the impact of covariates on the missingness, and admits any type (discrete or continuous) of variables. 
Second, under the stableness of resistance assumption we establish the semiparametric theory, which is for the first time in this field. 
We   characterize the tangent space, the efficient influence function and  the semiparametric efficiency bound for estimating  a general full-data functional. 
Third, we propose a suite of semiparametric estimation methods including an inverse probability weighted (IPW) estimator,   an outcome regression-based (REG) estimator, and a doubly robust (DR) one.
The proposed IPW estimator is a generalization of the calibration estimator of \cite{kim2014propensity} 
by allowing for nonlinear propensity score models and call-specific  effects of covariates on  nonresponse.
The IPW and REG estimators rest on correct specification of certain working models; otherwise, they are  no longer consistent.
However, the doubly robust estimator affords double protection against misspecification of working models: it remains consistent if   the  working models for either the IPW or the REG estimation is correct, but not necessarily both;
moreover, it  attains the efficiency bound for the nonparametric model when all working models are correct.
To  alleviate the concern about misspecification of parametric working models,
we further propose a more robust  approach  that uses more flexible working models to estimate nuisance functions and  plugs in the nuisance estimators into  the efficient influence function to obtain the estimator of the functional of interest.
This approach extends the methods for nonignorable missing data analysis by allowing for flexible estimation of the odds ratio function, instead of requiring it to be known or follow a parametric model.
We  show that the resultant estimator has second order bias, 
which affords more robust inference  provided that the nuisance estimators converge  sufficiently fast.
The characterization of  the bias involves a more sophisticated analysis rather than directly applying the second-order Taylor expansion technique.
This  technical contribution may shed some light on the analysis of complicated estimators.

The rest of this paper proceeds as follows. 
In Section 2, we further discuss the challenge for identification with callback data.
In Section 3, we characterize  the   stableness of resistance assumption  for using  callback data to identify the full-data distribution and establish the identification results.
In Sections 4 and 5,  we establish the semiparametric theory and describe   the IPW, REG, DR estimators, including one based on flexible working models.
We   show their asymptotic properties.
Section 6 includes extensions  to estimation of a general full-data functional and to the setting with multiple callbacks.
Sections 7 and 8 include numerical illustrations with   simulations and a real data application to  the Consumer Expenditure Survey \citep[CES;][]{ces2013measuring} dataset, respectively.
Section 9 concludes with  discussions on other extensions  and   limitations of the proposed approach.

%

\section{Preliminaries and Challenges to Identification}
Throughout the paper, we let $(X,Y)$ denote the frame variables that are   released  in a survey or observational study for investigation of  a particular scientific purpose,
where  $X$ denotes a vector (possibly a null set) of fully observed covariates and $Y$ the  outcome prone to missing values.
The frame variables $X$ and $Y$ can be either  continuous or discrete, either univariate or multivariate.
Suppose in the data collection process  the interviewer has  continued  to contact a unit until he/she responds or the  maximum number of call attempts  $K$ is achieved.
We let $R_k$ denote the availability status  of $Y$ after the $k$th call, with $R_k=1$ if $Y$ is available  after the $k$th call and $R_k=0$ otherwise.
By this definition, we have $R_{k+1}\geq R_k$;  because if $R_k=1$, i.e. $Y$ is already obtained in the $k$th call, then by    definition $R_{k+1}=1\geq R_k$.
The observations of  $(R_1,\ldots, R_K)$ recording the response process in the survey  are known as callback data.
This setting differs from the classical monotone missingness pattern and should not be viewed as a longitudinal setting, 
because the callback data represent the availability status of exactly the same outcome in different call attempts, while the monotone missingness or longitudinal setting concerns the missingness of different outcomes.
Let $\pr$  denote a generic probability density or mass function.
The full data can be viewed as  $n$ independent and identically distributed  samples  from the   joint distribution $\pr(X,Y,R_1,\ldots, R_K)$, 
and in  the observed data the values of $Y$ with $R_K=0$ are missing.
Table \ref{tbl:callback} illustrates the data structure of a  survey with callbacks,
where the final availability  status $R_K$   can be created   based on the availability status of frame data but the callback data $(R_1,\ldots, R_K)$ contain much more information about the  response process than solely $R_K$  does.
To ground ideas, we will focus on the setting with two calls or when only the first two calls are considered  for analysis and then  discuss the extension to multiple calls in Section 6 and  the supplement.
Our primary goal is to make inference about the outcome mean $\mu=E(Y)$ based on the observed data and extension to a general functional of the full-data distribution $\pr(X,Y)$ will be considered in Section 6.

\begin{table}[H]
\centering
\caption{Data structure of a   survey with callbacks; NA denotes missing values}\label{tbl:callback}
\begin{tabular}{ccccccccccc}
\multicolumn{3}{c}{Frame data}  & &\multicolumn{4}{c}{Callbacks}\\
ID&$X$&$Y$&& $R_1$&$R_2$&$\ldots$&$R_K$\\
1&$x_1$ & $y_1$ &&1&1&$\cdots$&1\\
2&$x_2$& $y_2$ &&0&1&$\cdots$&1\\
$\colon$&$\colon$& NA& & 0&0&$\cdots$&0\\
$\colon$&$\colon$& NA& & 0&0&$\cdots$&0\\
$n$&$x_n$& $y_{n}$&&  0&0&$\cdots$&1\\
\end{tabular}
\end{table}

%

Traditionally, only the final response status $R_K$ is used for missing data analysis.
Beyond $R_K$, callback data $(R_1,\ldots, R_{K-1})$   reflect the reluctance of units to respond and how the reluctance is related to the outcome, which is particularly  informative about the nonresponse process.
However, callback data still do not suffice for identification.
Consider a toy example of survey on a binary outcome $Y$ with two calls; the joint distribution $\pr(Y,R_1,R_2)$ contains six unknown parameters $\{\pr(Y,R_1=0,R_2),\pr(Y,R_1=1,R_2=1): Y=0,1 \text{ and } R_2=0,1\}$  with a constraint that they sum to one;
the observed-data distribution  is captured by $\pr(Y,R_1=1,R_2=1)$ and $\pr(Y,R_1=0, R_2=1)$  with $Y=0,1$, and only offers four constraints on the unknown parameters. 
These five constraints do not suffice to  identify $\pr(Y,R_1,R_2)$ without further restrictions.
To achieve identification with callback data, \cite{alho1990adjusting}, \cite{kim2014propensity}, and \cite{qin2014semiparametric} consider a common-slope linear logistic model   for the   propensity scores.
\begin{example}[A common-slope linear logistic model]\label{example:1}
\cite{alho1990adjusting} assumes that for $k=1,\ldots, K$, 
\begin{equation}\label{alho}
 \logit \pr(R_k=1\mid R_{k-1}=0, X,Y)= \alpha_{k0} + \alpha_{k1}X +\gamma_k Y,\quad
\text{with}\quad  \alpha_{k1}=\alpha_1,\quad \gamma_k=\gamma.
\end{equation}
\end{example}
Hereafter, we let $\logit x = \log\{x/(1-x)\}$  be the logit    transformation  of $x$ and   $\expit x = \exp(x)/\{1+\exp(x)\}$ its  inverse.
Example \ref{example:1}  is a   restrictive model: the   propensity score for each callback  follows 
a linear (in parameters) logistic model and        the effects of  both the covariate $X$ and the outcome $Y$ do not vary across  call attempts. 
\cite{daniels2015pattern}   mention an   identification result  that   allows for call-specific coefficients of $X$ but  requires constant coefficients of $Y$  across all call attempts. 
\cite{guan2018semiparametric} relax the common slope assumption by only  requiring    $\alpha_{11}=\alpha_{21}$ and  $\gamma_1=\gamma_{2}$.
This enables   identification with only two call attempts, but their identifying strategy based on a meticulous analysis of the logistic model only admits continuous covariates and outcome.
In order to make full use of callback data for nonresponse adjustment, it is   of great interest to establish identification for other types of models than the linear logistic model.
So far, however,  identification and inference of semiparametric and nonparametric propensity score models with callbacks are not available.

\section{Identification under a stableness of resistance assumption}
\subsection{The stableness of resistance assumption}

An all-important step for  identification under MNAR is  to characterize  the     degree to which the missingness departs from MAR.
The odds ratio function is a particularly useful measure widely used in the missing data literature to  characterize the  association between the outcome and response propensity \citep[e.g.][]{rotnitzky2001methods,osius2004association,chen2007semiparametric,kim2011semiparametric,miao2016varieties,franks2020flexible,malinsky2020semiparametric}.
We define the odds ratio functions for the response propensity in the  first and second calls as follows,
\begin{eqnarray}
\Gamma_1(X,Y) &=&\log \frac{\pr(R_1=1\mid X,Y)\pr(R_1=0\mid X,Y=0)}{\pr(R_1=0\mid X,Y)\pr(R_1=1\mid X,Y=0)}, \\
\Gamma_2(X,Y) &=& \log \frac{\pr(R_2=1\mid R_1=0,X,Y)\pr(R_2=0\mid R_1=0, X,Y=0)}{\pr(R_2=0\mid R_1=0, X,Y)\pr(R_2=1\mid R_1=0,X,Y=0)}. 
\end{eqnarray}
The odds ratio functions $\Gamma_1(X,Y)$ and  $\Gamma_2(X,Y)$ measure the impact of the outcome on the missingness   in the first and second calls, respectively, by  quantifying   the change in the odds of response  caused by a shift of $Y$ from a common reference  value  $Y=0$ to a specific level while controlling for covariates.
Probabilities $\pr(R_1=1\mid X,Y=0)$ and $\pr(R_2=1\mid R_1=0,X,Y=0)$ are referred to as  baseline propensity scores evaluated at the reference value $Y=0$.
Any other value within the support of $Y$ can  be chosen as the reference value.
Therefore, the odds ratio function is  a measure of the resistance to respond caused by the outcome.
A positive odds ratio function evaluated at $(x,y)$ indicates that for participants with $X = x$, those with $Y = y$ are more likely to respond than those with $Y = 0$, 
and a larger odds ratio means a stronger willingness to respond. 
The setting with odds ratio functions equal to zero corresponds to MAR where the outcome does not influence the  response. 
Our identification strategy rests on the following assumption.
\begin{assumption}\label{assump:odds}
\begin{itemize}
\item[(i)] Stableness of resistance: $\Gamma_1(X,Y)=\Gamma_2(X,Y)=\Gamma(X,Y)$ for some   $\Gamma(X,Y)$;
\item[(ii)] Positivity: $0<\pr(R_1=1\mid X,Y) <1$ and $0<\pr(R_2=1\mid R_1=0,X,Y) <1$ for all $(X,Y)$.
\end{itemize}
\end{assumption}
Condition (ii) is a standard positivity assumption   in missing data analysis, which ensures  sufficient  overlap between  nonrespondents and respondents  in each call attempt.
Condition (i) is the key identifying condition, which   reveals   that  the impact of the outcome on   the response propensity   remains the same in  the first two calls. 
Whether Condition  (i) holds or not does not depend on the choice of  the reference value in the definition of the odds ratio functions.
Although Assumption \ref{assump:odds} is   concerned with  the missingness process, it in fact imposes certain sophisticated restrictions on the outcome distribution.
This is because under MNAR the missingness process is not ancillary to the outcome distribution; 
see Lemma \ref{lemma:1} in the supplement for   the  apparent dependence between the missingness process, the outcome distribution and the odds ratio function.

Similar to various missing data problems, Assumption \ref{assump:odds}   is untestable based on observed data. 
Thus, the justification of  its validity   requires domain-specific knowledge and needs to be investigated on a case-by-case basis.
Assumption \ref{assump:odds} is motivated from the idea that the   individual characteristics may  influence the missingness process in  the same way across the survey contacts. 
Similar ideas  have been used since \cite{alho1990adjusting}, \cite{kim2014propensity}, \cite{qin2014semiparametric}, and \cite{guan2018semiparametric}.
Assumption \ref{assump:odds} is plausible if the nonignorable missingness  is  caused by certain social stigma  (or other issues such as  difficulty to reach)  and the social stigma attached to the units does not change during the first two calls or the  contacts are made in   a short period of time.
Besides,  sensitivity analysis can be applied to assess how results would change if Assumption  \ref{assump:odds} were to be violated; see the simulation and application in Sections 7 and 8, respectively.

We provide a  generative model  based on the discrete choice theory \citep{mcfadden2001economic,train2009discrete} to further illustrate  the intuition behind   Assumption \ref{assump:odds}. 
Suppose
\begin{eqnarray*}
	&&Y = \beta_0 + \beta_1X + C, \\
	&&U_{1} = \bar\alpha_{10} + \bar\alpha_{11}X + \gamma_1 C + \varepsilon_{1},\quad
	U_{2} = \bar\alpha_{20} + \bar\alpha_{21}X + \gamma_2 C + \varepsilon_{2},\\
	&&R_{1} = I(U_{1} > 0), \quad
	R_{2} = R_1 + (1-R_1)\times I(U_{2} > 0),
\end{eqnarray*}
where $X$ is the observed covariate and $C$ is an unmeasured factor (such as social stigma)  correlated with both the outcome and missingness mechanism.
Variables $U_{1}$ and $U_{2}$   are  the utilities (or net benefits) that a subject obtains from responding as opposed to not responding at the first and second calls, respectively, 
and $\varepsilon_{1}$ and $\varepsilon_{2}$ are independent error terms, both  following  a logistic distribution and   independent of $X$ and $C$. 
The generating process of $R_1,R_2$ reveals how the subjects decide to respond:
a subject will respond in the first call (i.e., $R_{1}=1$) if   utility $U_1>0$;
for those not responded in   the first call, he/she will respond at the second call (i.e., $R_1=0,R_{2}=1$) if    utility $U_2>0$.
We show in Section S1 of the supplement that the above  discrete choice model  induces the linear logistic propensity score  model:  
$\logit \pr(R_k=1\mid R_{k-1}=0, X,Y)= \alpha_{k0} + \alpha_{k1}X +\gamma_k Y$, 
with $\alpha_{k0} = \bar\alpha_{k0} - \gamma_k\beta_0$ and $\alpha_{k1} = \bar\alpha_{k1} - \gamma_k\beta_1$  for $k=1,2$.
In this model, the stableness of resistance stating that $\gamma_1=\gamma_2$ in fact  means the effects of the unmeasured factor $C$ on the utilities $U_1$ and   $U_2$ remain the same.
However, the stableness of resistance does not require $\bar\alpha_{11}$ and $\bar\alpha_{21}$  to be equal, which admits varying effects of the observed covariate $X$ on the utilities and  the response in different calls.

The biggest difference between Assumption \ref{assump:odds} and Model \eqref{alho} is that,
the former  is about social reality that researchers can justify based on domain-specific knowledge, 
whereas the latter further depends on   functional form restrictions that may lack real-world justification.
Besides, Assumption \ref{assump:odds}  has several major differences from Model \eqref{alho}, while including the latter as a special case.
Under Assumption \ref{assump:odds}  the propensity scores have the following representation:
\begin{eqnarray}
        \pr(R_1=1\mid X,Y) &=&\expit\{A_1(X) + \Gamma(X,Y)\},\label{propensity1}\\
        \pr(R_2=1\mid R_1=0,X,Y)&=&  \expit\{A_2(X) + \Gamma(X,Y)\}, \label{propensity2}
\end{eqnarray}
where $A_1(X)= \logit \pr(R_1=1\mid X,Y=0)$ and $A_2(X) = \logit \pr(R_2=1\mid R_1=0,X,Y=0)$ are the logit transformations of the corresponding  baseline propensity scores which  capture  impacts of covariates  on the missingness,
and $\Gamma(X,Y)$ is the common odds ratio function that measures the impact of the outcome on the missingness  in the first and second calls.
The functions $A_1(X), A_2(X), \Gamma(X,Y)$  are unrestricted.
Therefore, Assumption \ref{assump:odds}   admits nonlinear and nonparametric effects of covariates and missing outcomes on the missingness, 
which is a much larger, and in fact,  infinite-dimensional  model;
it accommodates any types of outcomes.
The assessment of   Assumption \ref{assump:odds} only concerns   the first two calls and does not depend on the number of callbacks, while the missingness in the remaining calls are left completely unrestricted;
it only requires  the effects of missing outcomes, captured by the odds ratio, to be stable, while admiting call-specific impacts of   covariates on nonresponse.
In contrast, Model \eqref{alho} is fully parametric  and the effects of covariates and missing outcomes on the missingness must be linear,
which is a  restrictive model indexed by several parameters.
Model \eqref{alho} is a special case  of Assumption \ref{assump:odds} where $A_1(X) = \alpha_{10} + \alpha_1^\t X$, $A_2(X) = \alpha_{20} + \alpha_1^\t X $, and $\Gamma(X,Y)=\gamma Y$ encode the effects of variables on the missingness with the regression coefficients.  
This is best suited for  continuous or binary outcomes, but may not be suitable   for other types of outcomes.
Model \eqref{alho}  imposes   parametric functional restrictions on  all contact attempts and  assumes  stable effects of all variables entering the model;
however, this  may not hold for the entire time range when $K$ is large or callback data are from a long period.
In short,  Model \eqref{alho}  is a  restrictive  model making strong  parametric functional restrictions beyond the parameter (e.g. the outcome mean) of interest.
Assumption \ref{assump:odds} makes much  less restrictive  assumptions  and has much greater flexibility.

\subsection{Nonparametric identification}

We will show  that Assumption \ref{assump:odds} suffices for identification of the joint distribution $\pr(X,Y,R_1,R_2)$.
We first briefly explain  why callback data are useful and how to leverage them for identification.
Consider identification of  $f(Y,R_1\mid X)$, which  hinges on   identification  of $f(Y, R_1=0\mid X)$, i.e., the  missing data distribution in the first  call.
In the presence of nonignorable missing data,  we need  to determine the selection bias, captured by 
\[\frac{f(Y,R_1=1\mid X)}{f(Y,R_1=0\mid  X)}=\frac{\pr(R_1=1\mid X,Y)}{\pr(R_1=0\mid X,Y)}.\]
With callback data, a natural idea for estimating this selection bias is to approximate   $f(Y,R_1=0\mid X)$ with $f(Y, R_1=0,R_2=1\mid X)$, where $f(Y, R_1=0,R_2=1\mid X)$  is the distribution of  the    data  obtained in the second call.
This is analogous to the two-phase sampling \citep{deming1953probability}.
However,  the challenge here is that    nonresponse    still occurs in the second call and is possibly not at random,
and as a result, $f(Y, R_1=0,R_2=1\mid X)$ is not necessarily equal to $f(Y, R_1=0\mid X)$ and this crude approximation does not suffice for identification.
Nonetheless, under   Assumption 1 we are able to  characterize  and can further identify the bias  of this approximation as shown in Proposition \ref{prop:idn} and Theorem \ref{thm:idn}  below.

\begin{proposition}\label{prop:idn}
Letting $D(X)=A_2(X)-A_1(X)$, then under Assumption \ref{assump:odds} we have that 
\begin{eqnarray*}
\frac{\pr(R_1=1\mid X,Y)}{\pr(R_1=0\mid X,Y)} &=&  \frac{\pr(Y, R_1=1\mid X)}{\pr(Y,R_2=1,R_1=0\mid X)} - \exp\{-D(X)\},\\
\frac{\pr(R_2=1\mid R_1=0,X,Y)}{\pr(R_2=0\mid R_1=0,X,Y)}  &=& \exp\{D(X)\}  \frac{\pr(Y, R_1=1\mid X)}{\pr(Y,R_2=1,R_1=0\mid X)} - 1.
\end{eqnarray*}
\end{proposition}
We prove this result in the supplement. 
Proposition \ref{prop:idn} also implies  an inequality:
\begin{equation}\label{restriction}
\frac{\pr(Y, R_1=1\mid X)}{ \pr(Y,R_2=1,R_1=0\mid X)}> \exp\{-D(X)\},
\end{equation}
i.e., given $X$ the density ratio on the left hand side is uniformly bounded from zero. 
This is a  restriction on the observed-data distribution imposed by  Assumption \ref{assump:odds}.

In  Proposition \ref{prop:idn}, $\pr(Y, R_1=1\mid X)/\pr(Y,R_2=1,R_1=0\mid X)$ is used to approximate $\pr(R_1=1\mid X,Y)/\pr(R_1=0\mid X,Y)$ and 
the approximation bias is captured by   $D(X)$.
We can further show that $D(X)$ is identified under Assumption \ref{assump:odds}, and thus the selection bias $\pr(R_1=1\mid X,Y)/\pr(R_1=0\mid X,Y)$ is identified,  and then  the identification of $f(Y,R_1\mid X)$ and $\pr(X,Y,R_1,R_2)$ is straightforward.

\begin{theorem}\label{thm:idn}
Under Assumption \ref{assump:odds}, $D(X)$ is identified, and as a result, $\pr(X,Y,R_1,R_2)$ is identified from the observed-data distribution.
\end{theorem}
\begin{proof}
Proposition \ref{prop:idn} implies   that
\begin{eqnarray}
\frac{\pr(R_2=0\mid R_1=0, X)}{\pr(R_2=1\mid R_1=0, X)} &=& \int \frac{\pr(R_2=0\mid R_1=0,X,Y)}{\pr(R_2=1\mid R_1=0,X,Y)} \pr(Y\mid R_2=1,R_1=0, X)dY \nonumber \\
 &=&\int  \left[ \frac{\exp\{D(X)\}  \cdot \pr(Y, R_1=1\mid X)}{\pr(Y,R_2=1,R_1=0\mid X)} - 1\right]^{-1} \pr(Y\mid R_2=1,R_1=0,X) dY   \nonumber\\
 &\equiv & L\{D(X)\}.  \label{idn3}
\end{eqnarray}
This is an equation with  $D(X)$ unknown  while all the other quantities   are available from the observed-data distribution.  
Identification of $D(X)$ can be assessed by checking uniqueness of  the solution to this equation.
For any fixed $x$ and any $D(x)$ such that \eqref{restriction} is satisfied,    $L\{D(x)\}$ is strictly decreasing in $D(x)$ because
\begin{eqnarray*}
\frac{\partial L\{D(x)\}}{\partial D(x)} = - \int \frac{\exp\{D(x)\}   \frac{\pr(Y, R_1=1\mid X=x)}{\pr(Y,R_2=1,R_1=0\mid X=x)}}{ \left[\exp\{D(x)\}  \frac{\pr(Y, R_1=1\mid X=x)}{\pr(Y,R_2=1,R_1=0\mid X=x)} - 1\right]^2} \pr(Y\mid R_2=1,R_1=0, X=x) dY <0.
\end{eqnarray*}
Therefore,  for any fixed $x$ the solution to \eqref{idn3} is unique. 
Applying this argument to all $x$,  then $D(X)$ is identified and $\pr(R_1=1\mid X,Y),\pr(R_2=1\mid R_1=0,X,Y)$ are identified according to Proposition \ref{prop:idn}.
Then it is straightforward to show that  $\pr(X,Y)$ and  $\pr(X,Y,R_1,R_2)$ are identified.
\end{proof}

We achieve nonparametric identification of   $\pr(X,Y,R_1,R_2)$  with callback data under the stableness of resistance  assumption.
The nonparametric identification elucidates the underlying source   for identification with callback data, other than invoking parametric functional restrictions; 
it extends the application of callback data and opens the way to  novel  estimation methods.
To our knowledge,  Assumption \ref{assump:odds} is so far the most parsimonious condition characterizing the most
flexible model for identification  with callbacks, 
and  Theorem \ref{thm:idn} is so far the most general  identification result.
Back to the binary outcome example, Assumption \ref{assump:odds} can be viewed as an additional constraint on the parameters so that we have sufficient number of constraints to identify the joint distribution;
see Section S1 in the supplement for the details about the parameter count interpretation of identification.
Note that  Assumption \ref{assump:odds} is sufficient but not necessary for identification of the joint distribution. 
There  exist different  conditions that can achieve identification.
For instance, in a similar spirit, it is of interest and  possible to establish identification with stableness of resistance on other scales of the propensity scores, say risk ratio scale, because it also introduces certain constraints on the parameters.
In principle, the union of all such conditions constitutes the if and only if condition for identification.
However, it may not be worth the chase because such an if and only if  condition will be very complicated with difficulty in interpretation and justification in practice, and is thus less preferred in contrast to a sufficient and meaningful condition like the one we propose in the paper.
Applying Theorem \ref{thm:idn} to the linear logistic model immediately gives  the following result.

\begin{proposition}\label{prop:idn2}
Assuming that 
$\logit \pr(R_k=1\mid R_{k-1}=0, X,Y)= \alpha_{k0} + \alpha_{k1} X +\gamma Y$ for $k=1,2$, 
then $\alpha_{k0},\alpha_{k1}$, and $\gamma$ are identified.
\end{proposition}
Proposition \ref{prop:idn2}  generalizes the identification of Model \eqref{alho}   by admitting   call-specific coefficients of $X$ in the propensity scores and only requiring the coefficients of $Y$ to be equal for the first two calls.
Model \eqref{alho} also  reveals a continuum-of-resistance model where
nonrespondents are more similar to   late respondents than to early ones due to a common odds ratio for all call attempts.
However, in the model in Proposition \ref{prop:idn2}   nonrespondents may depart further from late respondents than from early ones, depending on the form of    odds ratio functions in later calls.
In this case,  continuum-of-resistance models are not suitable.

For estimation, we can in principle first estimate   $\pr(Y\mid R_2=1, R_1=0, X)$, $\pr(Y,R_2=1, R_1=0\mid X), \pr(Y,R_1=1\mid X)$, and $\pr(R_2=1\mid R_1=0, X)$, then plug   them  into equation \eqref{idn3}
to solve for $D(X)$, and finally obtain the estimate of $\pr(X,Y,R_1,R_2)$ according to Proposition \ref{prop:idn}. 
The first step can be achieved by standard nonparametric estimation and the third step only involves basic 
arithmetic; 
however, solving equation \eqref{idn3} in the second step  is in general complicated.
In the next section, we consider a feasible parametrization for the joint distribution $\pr(X,Y,R_1,R_2)$ and develop  estimation methods.

\section{Inverse probability weighted and outcome regression-based estimation}
\subsection{Parametrization}
Under Assumption \ref{assump:odds}, 
we introduce the following factorization of the joint distribution as the basis for parametrization and estimation.
\begin{eqnarray}
\pr(Y,R_1,R_2\mid X)&=& c_1(X)\cdot  \pr(R_1\mid X,Y =0) \cdot \exp\{(R_1-1)\Gamma(X,Y)\} \cdot \pr(Y \mid R_2=1, R_1=0,X) \nonumber\\
&&\cdot   \{\pr(R_2\mid R_1=0,X,Y)\}^{1-R_1} \cdot \{\pr(R_2=0\mid R_1=0,X,Y)\}^{-1}, \label{factor1}\\
c_1(X) &=& \frac{\pr(R_1=1\mid X)}{\pr(R_1=1\mid X,Y=0)}\frac{1}{E\{1/\pr(R_2=0\mid R_1=0,X,Y)\mid R_2=1, R_1=0, X\}},\nonumber
\end{eqnarray}
where $c_1(X)$ is a  normalizing function of $X$ that makes the right hand side  of \eqref{factor1} a valid density function.
We prove \eqref{factor1} in Lemma \ref{lemma:1} in the supplement.
This factorization enables a convenient and congenial specification of four  components of the joint distribution:
\begin{itemize}
\item two  baseline propensity scores   $\pr(R_1=1\mid X,Y=0)$ and $\pr(R_2=1\mid R_1=0, X,Y=0)$;
\item the odds ratio function $\Gamma(X,Y)$;
\item and the   outcome distribution  for the second call $\pr(Y\mid R_2=1,R_1=0,X)$.
\end{itemize}
This kind of factorization of a joint density into  a combination of univariate conditionals and   odds ratio
is widely applicable; see  \cite{osius2004association}, \cite{chen2007semiparametric}, \cite{kim2011semiparametric},  and \cite{franks2020flexible} for examples in missing data analysis and causal inference. 
For notational convenience, we let   $\pi_1(X,Y)=\pr(R_1=1\mid X,Y)$ and $\pi_2(X,Y)=\pr(R_2=1\mid R_1=0,X,Y)$ denote the propensity scores,
$A_1(X)= \logit \pr(R_1=1\mid X,Y=0)$ and  $A_2(X) = \logit \pr(R_2=1\mid R_1=0,X,Y=0)$  the logit transformations of the corresponding baseline propensity scores, 
and $\pr_2(Y\mid X) =\pr(Y\mid R_2=1,R_1=0,X)$ the   outcome model for the second call.
Note that   $\pi_1(X,Y)$ and $\pi_2(X,Y)$ are determined by $A_1(X), A_2(X), \Gamma(X,Y)$ as in equations \eqref{propensity1}--\eqref{propensity2} and that the joint distribution is determined once given $A_1,A_2,\Gamma,f_2$.
We will write  $\pi_1, \pi_2, A_1, A_2, \pr_2,\Gamma$ for short where it does not cause confusion.

One may have concern  that  inequality \eqref{restriction} makes the range of $\pr(R_2 = 1\mid R_1=0,X,Y)$ dependent on the value of $\pr(R_1= 1\mid X,Y)$.
However,  we do not model these two propensity scores  separately,
and  we model  $A_1(X), A_2(X)$ and $\Gamma(X,Y)$ instead. 
We show in Lemma 1 in the supplement that 
these four components $A_1, A_2, \pr_2,\Gamma$ of the joint distribution are indeed variationally independent.
In the next, we consider semiparametric estimation under correct specification of a subset of these four models.
If one would like to consider a different   parametrization of the joint distribution,   
we refer to \citet{richardson2017modeling} about how to address
the issue of variational dependence.

\subsection{Inverse probability weighting}
We specify   parametric working models for two baseline propensity scores $A_1(X;\alpha_1)$, $A_2(X;\alpha_2)$, and the odds ratio function $\Gamma(X,Y;\gamma)$.
By definition, we require $\Gamma(X,Y=0;\gamma)=0$.
This  is equivalent to specifying propensity score models   $\pi_1(X,Y;\alpha_1,\gamma)$ and $\pi_2(X,Y;\alpha_2,\gamma)$.
The logistic model in Proposition \ref{prop:idn2} is   an example.
The following  equations   characterize the propensity scores,
\begin{eqnarray}
0&=&E\left\{\frac{R_1}{\pi_1}-1\mid X,Y\right\},\label{ipw0.1}\\
0&=&E\left\{\frac{R_2 - R_1}{ \pi_2} - (1-R_1) \mid X,Y\right\}, \label{ipw0.2}\\
0&=&E\left\{\frac{R_2 - R_1}{ \pi_2} - \frac{1-\pi_1}{\pi_1}R_1 \mid X,Y\right\}. \label{ipw0.3}
\end{eqnarray}
The first equation follows from the definition of $\pi_1$ and the other two   echo the definition of $\pi_2$  by noting that 
$E\left\{(R_2 - R_1)/\pi_2 - (1-R_1) \mid X,Y\right\}=E\{(1-R_1) (R_2/\pi_2 - 1)\mid X,Y\}=0$.
These three conditional moment equations motivate  the following marginal moment equations for estimating   $(\alpha_1,\alpha_2,\gamma)$:
\begin{eqnarray}
0&=&\hat E\left[\left\{\frac{R_1}{\pi_1(\alpha_1,\gamma)}-1\right\}  \cdot V_1(X)\right],\label{ipw1.1}\\
0&=&\hat E\left[\left\{\frac{R_2-R_1}{\pi_2(\alpha_2,\gamma)}-(1-R_1)\right\}\cdot V_2(X) \right],\label{ipw1.2}\\
0&=&\hat  E\left[\left\{\frac{R_2 - R_1}{ \pi_2(\alpha_2,\gamma)}  -  \frac{1-\pi_1(\alpha_1,\gamma)}{\pi_1(\alpha_1,\gamma)}R_1 \right\}\cdot U(X,Y)\right], \label{ipw1.3}
\end{eqnarray}
where $\hat E$ denotes the empirical mean operator and  $V_1(X)=\partial A_1(X;\alpha_1)/\partial \alpha_1$, $V_2(X)=\partial A_2(X;\alpha_2)/\partial \alpha_2$, $U(X,Y)=\partial \Gamma(X,Y;\gamma)/\partial \gamma$.
For instance, in the linear logistic model with $A_1(X;\alpha_1)=(1,X^\t)\alpha_1$, $A_2(X;\alpha_2)=(1,X^\t)\alpha_2$, $\Gamma(X,Y;\gamma)=Y \gamma$, 
one may use $V_1(X)=V_2(X)=(1,X^\t)^\t,U(X,Y)=Y$.
Note that   $V_1,V_2,U$ can be chosen as other user-specified functions; see   \citet[page 30]{tsiatis2006semiparametric}.

Equations  \eqref{ipw1.1}--\eqref{ipw1.3} only involve the observed data.
The generalized method of moments \citep{hansen1982large} can be implemented to solve these equations.
Letting $(\hat\alpha_{1,\ipw},\hat\alpha_{2,\ipw},\hat\gamma_\ipw)$ be the   nuisance estimators obtained from \eqref{ipw1.1}--\eqref{ipw1.3}, $\hat\pi_1=\pi_1(\hat\alpha_{1,\ipw},\hat\gamma_\ipw)$ and $\hat\pi_2=\pi_2(\hat\alpha_{2,\ipw},\hat\gamma_\ipw)$ the estimated propensity scores, and $\hat p_2 =  \hat\pi_1 +  \hat\pi_2 (1-  \hat\pi_1 )$ an estimator of $p_2=\pr(R_2=1\mid X,Y)$,
we propose the following  IPW estimator of the outcome mean,
\begin{eqnarray}
\hat \mu_\ipw = \hat E\left\{\frac{R_2}{\hat p_2 } Y\right\}. \label{esti:ipw1}
\end{eqnarray}

Preceding our proposal, \cite{kim2014propensity} developed a calibration estimator   under the  common-slope logistic model \eqref{alho}.
Our IPW estimator can be viewed as a generalization of the calibration estimator,
which can  in fact  be obtained from our IPW estimator by  choosing   appropriate functions $V_1,V_2,U$.
However, the calibration   estimator  may be biased  if the slopes of covariates $X$ differ in the two propensity score models, 
while our IPW estimator  works in this case.
We include a numerical simulation comparing these two estimators in Section S6.1 of the supplement.
Under the common-slope linear logistic model, \cite{qin2014semiparametric} and \cite{guan2018semiparametric} developed  empirical likelihood-based  estimation that is convenient to   incorporate auxiliary information to achieve  higher efficiency; it is of   interest to extend their approach to the setting with call-specific  slopes for covariates.

\subsection{Outcome regression-based estimation}
Alternatively,   the outcome mean can be obtained by estimation or imputation   of the missing values.
We specify and fit working models $\pi_1(\alpha_1,\gamma), \pr_2(Y\mid X;\beta)$ for    the first-call propensity score $\pr(R_1=1\mid X,Y)$ and the second-call outcome   distribution $\pr(Y\mid X,R_2=1,R_1=0)$, 
and  impute    the missing values  with
\begin{eqnarray}
\pr(Y\mid X,R_2=0) &=&  \frac{\exp(-\Gamma)\pr(Y\mid X,R_2=1,R_1=0)}{E\{ \exp(-\Gamma)\mid X, R_2=1,R_1=0\}}, \label{tilting} 
\end{eqnarray}
which is known as the exponential tilting  or Tukey’s representation \citep{kim2011semiparametric,vansteelandt2007estimation,franks2020flexible}.
To obtain  the nuisance  estimators $(\hat \beta, \hat \alpha_{1,\reg},\hat \gamma_\reg)$, we solve 
\begin{eqnarray}
0 &=& \hat{E} \left\{ (R_2-R_1) \cdot \frac{ \partial \log f_2(Y\mid X;\beta)}{\partial \beta} \right\},\label{reg1.1} \\
0 &=& \hat E\left[ \left\{\frac{R_1}{\pi_1(\alpha_1,\gamma)} -R_2\right\} U(X,Y) - (1-R_2) E\{U(X,Y) \mid X, R_2=0;\beta,\gamma\}  \right],\label{reg1.2}
\end{eqnarray} 
where $E(\cdot  \mid X, R_2=0;\beta,\gamma)$ is evaluated according to \eqref{tilting} and $U(X,Y)=\{\partial A_1(X;\alpha_1)/\partial \alpha_1,\partial \Gamma(X,Y;\gamma)/\partial \gamma\}$. 
Note that   $U(X,Y)$ can be chosen as other user-specified functions; see      \citet[page 30]{tsiatis2006semiparametric}.
Equation \eqref{reg1.1} is  the score equation of $\pr_2(Y\mid X;\beta)$.
Equation \eqref{reg1.2}  is motivated by  the  equation  $\hat E[(1-R_2)\{U(X,Y) -E\{U(X,Y)\mid R_2=0,X\}\}]=0$,
but  the  evaluation of  $\hat E\{U(X,Y)\}$   is untenable due to missing values of $Y$,  
and thus we replace with $\hat E\{R_1 U(X,Y)/\pi_1\}$.
An   estimator of the outcome mean by  imputing the missing outcome values  with $E(Y\mid   X,R_2=0;\hat\beta,\hat\gamma_\reg)$ is
\begin{equation}\label{esti:reg1}
\hat \mu_\reg =\hat E \{R_2 Y + (1-R_2)  E(Y\mid X,R_2=0;\hat\beta,\hat\gamma_\reg)\}.
\end{equation}
We refer to this approach as the outcome regression-based  (REG) estimation, 
although the  first-call propensity score model is also involved.

\cite{guan2018semiparametric} have previously proposed an estimator that involves the full-data outcome regression $E(Y\mid X)$  whereas our REG estimation   involves $E(Y\mid R_2=1,R_1=0,X)$ that
only concerns the observed data  and  enjoys the ease for model specification and estimation.
Moreover,     estimation of their outcome regression model requires    estimating the propensity scores for all call attempts, but  our REG estimation    only requires estimating   the first-call propensity score.

Note that \eqref{ipw1.1}--\eqref{esti:ipw1} are unbiased estimating equations for $(\alpha_1,\alpha_2,\gamma,\mu)$ in   model
\[\M_\ipw=\{\pr(X,Y,R_1,R_2): 
\text{ Assumption \ref{assump:odds} holds; } \pi_1(\alpha_1,\gamma) \text{ and  }\pi_2(\alpha_2,\gamma) \text{ are correct}\},\] 
and  \eqref{reg1.1}--\eqref{esti:reg1} are unbiased estimating equations for $(\beta,\alpha_1,\gamma,\mu)$  in model
\[\M_\reg=\{\pr(X,Y,R_1,R_2): 
\text{ Assumption \ref{assump:odds} holds; } \pi_1(\alpha_1,\gamma) \text{ and  }\pr_2(Y\mid X;\beta) \text{ are correct}\}.\] 
As a consequence,  consistency and asymptotic normality  of $(\hat\alpha_{1,\ipw},\hat\alpha_{2,\ipw},\hat\gamma_\ipw,\hat\mu_\ipw)$ in model $\M_\ipw$ and of $(\hat\beta, \hat\alpha_{1,\reg}, \hat\gamma_\reg,\hat\mu_\reg)$ in  model $\M_\reg$ can be established under standard regularity conditions \citep[see e.g.,][]{newey1994large} by following the  theory of estimating equations,  which we will not replicate here.
The choice of user-specified functions $V_1(X),V_2(X),U(X,Y)$ depends on the working models and influences the efficiency of the   estimators.
In principle,  the optimal choice can be obtained  by deriving the  efficient influence function for the nuisance parameters in models $\M_\ipw$ and $\M_\reg$, respectively.
However,  the potential prize of attempting to attain  local efficiency for nuisance parameters may not always be worth the chase because it depends on correct modeling of additional components of the  joint distribution beyond the working models, 
and as pointed out by \cite{stephens2014locally} such additional modeling efforts    seldom deliver the anticipated efficiency gain. 
Besides,    consistency of the estimators is undermined  if the required working models are incorrect.
Therefore, we  next propose       a doubly robust  and locally efficient approach.

\section{Semiparametric   theory and doubly robust estimation}
\subsection{Semiparametric theory}

We consider the  model $\M_\npr$ characterized by Assumption \ref{assump:odds},
\begin{equation*}
\M_\npr=\{\pr(X,Y,R_1,R_2): 
\text{ Assumption \ref{assump:odds} holds;   $A_1(X), A_2(X),\Gamma(X,Y),\pr_2(Y\mid X)$ are unrestricted.}
\}
\end{equation*}
Although the stableness of resistance assumption does impose an inequality constraint \eqref{restriction} on the observed-data distribution, we refer to $\M_\npr$  as the  (locally) nonparametric model because no parametric models are imposed in  $\M_\npr$ and as   established in the following,  the observed-data tangent space under  $\M_\npr$ is   the entire Hilbert space.
We aim to  derive the set of influence functions   for all regular and asymptotically linear (RAL) estimators of $\mu$ and  the efficient influence function under   $\M_\npr$.
To achieve this goal, the primary step is to  derive the observed-data tangent space for model $\M_\npr$.
Hereafter, we   let $O=(X,Y,R_1,R_2)$ if $R_2=1$ and $O=(X,R_1,R_2)$ if $R_2=0$   denote the observed data.
\begin{proposition}\label{prop:tangent}
The observed-data tangent space for $\M_\npr$ is 
\begin{eqnarray*}
\mathcal{T}=\{h(O)=R_1h_1(X,Y)+R_2h_2(X,Y)+h_3(X) ; E\{h(O)\}=0, E\{h^2(O)\}<\infty  \},
\end{eqnarray*}
where $h_1(X,Y)$, $h_2(X,Y)$ and $h_3(X)$ are arbitrary measurable and square-integrable functions.
\end{proposition}
This   proposition  states that the observed-data tangent space for $\M_\npr$ is   the entire Hilbert space   of  observed-data functions with mean zero and finite variance,  equipped with the usual inner product.
Hence,    the stableness of resistance assumption does not impose any local restriction on the observed-data  distribution.
As a result, there exits a unique influence function for $\mu$ in model $\M_\npr$, which must be the efficient one.
We have derived  the closed form for the efficient influence function. 
\begin{theorem}\label{thm:eif}
The efficient influence function for $\mu$ in the nonparametric model $\M_\npr$ is
\begin{equation*}\label{eif1}
\IF(O;\mu)=\left\{\frac{R_2-R_1}{\pi_2^2} - \frac{1-\pi_1}{\pi_1}\frac{R_1}{\pi_2}+ \frac{R_1}{\pi_1}  \right\}  \left\{Y -  \frac{E(Y/\pi_2  \mid X,R_2=0)}{E(1 /\pi_2\mid X,R_2=0)} \right\}  
+  \frac{E(Y/\pi_2 \mid X,R_2=0)}{E(1/\pi_2\mid X,R_2=0)} - \mu.
\end{equation*}
\end{theorem}
If we know a priori that  the missingness mechanism is at random, i.e., in the submodel of $\M_\npr$ with $\Gamma(X,Y)=0$,
or equivalently $R_1\ind Y\mid (X,R_2)$ and $R_2\ind Y\mid (X,R_1)$,  
we can show that the efficient influence function is 
\[\IF_\aipw (O;\mu)=  \frac{R_2}{\pr(R_2=1\mid X)} Y -  \left\{\frac{R_2}{\pr(R_2=1\mid X)} - 1\right\}E(Y\mid X)  - \mu.\]
\begin{proposition}\label{prop:mar}
 $\IF_\aipw(O;\mu)$ is the efficient influence function for $\mu$ in model  $\M_\mar =\{\pr(X,Y,R_1,R_2):\pr(X,Y,R_1,R_2) \in \M_\npr \text{ and } \Gamma(X,Y)=0\}$.
\end{proposition}
The tangent space of $\M_\mar$  is no longer the entire Hilbert space and Proposition 4 does not contradict Theorem 2.
It is well known that $\IF_\aipw (O;\mu)$ is the efficient influence function for $\mu$ in   the  model $\{f(X,Y,R_2): R_2\ind Y\mid X\}$ when the final availability   $R_2$ is MAR.
The availability of callback data $R_1$ does not change this efficiency bound when missingness in both calls are at random.

From Theorem \ref{thm:eif},   the semiparametric efficiency bound for estimating $\mu$ in $\M_\npr$ is $E\{ \IF(O;\mu)^2\}$. 
A locally efficient estimator  attaining this  semiparametric efficiency bound can be constructed by  plugging  
nuisance estimators  of   $\pi_1, \pi_2, E(\cdot\mid R_2=0,X)$ that converge sufficiently fast  into the efficient influence function and then solving $\hat E\{\IF(O;\mu)\}=0$.
However, one may not be confident that these nuisance parameters can  be modeled correctly.
It is therefore of interest to develop a doubly robust estimation approach, which     delivers valid inferences about the outcome mean provided that a subset  but not necessarily all low dimensional models  for the nuisance parameters are specified correctly. 

\subsection{Doubly robust estimation}
To construct a doubly robust (DR) estimator, 
we specify working models $\{A_1(X;\alpha_1), A_2(X;\alpha_2)$, $\Gamma(X,Y;\gamma)$, $\pr_2(Y\mid X;\beta)\}$ 
and estimate the nuisance parameters by solving 
\begin{eqnarray}
0 &=& \hat{E} \left\{ (R_2-R_1) \cdot \frac{ \partial \log f_2(Y\mid X;\beta)}{\partial \beta} \right\},\label{dr1.1} \\
0&=&\hat E\left[\left\{\frac{R_1}{\pi_1(\alpha_1,\gamma)}-1\right\}\cdot V_1(X) \right],\label{dr1.2}\\
0&=&\hat E\left[\left\{\frac{R_2-R_1}{\pi_2(\alpha_2,\gamma)}-(1-R_1)\right\}\cdot V_2(X) \right],\label{dr1.3}\\
0&=&\hat  E\left[\left\{\frac{R_2 - R_1}{ \pi_2(\alpha_2,\gamma)} -   \frac{1-\pi_1(\alpha_1,\gamma)}{\pi_1(\alpha_1,\gamma)}R_1 \right\}\cdot 
\left\{ U(X,Y) -  E(U(X,Y)\mid X, R_2=0;\beta,\gamma) \right\}\right], \label{dr1.4}
\end{eqnarray}
where $V_1(X)=\partial A_1(X;\alpha_1)/\partial \alpha_1$, $V_2(X)=\partial A_2(X;\alpha_2)/\partial \alpha_2$, $U(X,Y)=\partial \Gamma(X,Y;\gamma)/\partial \gamma$.
Note that   $V_1,V_2,U$ can be chosen as other user-specified functions; see discussions in Section 4.3 and    \citet[page 30]{tsiatis2006semiparametric}.
Let $(\hat\beta, \hat\alpha_{1,\dr},\hat\alpha_{2,\dr},\hat\gamma_\dr)$  denote the solution to \eqref{dr1.1}--\eqref{dr1.4}
and $\hat \pi_1=\pi_1(\hat\alpha_{1,\dr},\hat\gamma_\dr), \hat \pi_2=\pi_2(\hat\alpha_{2,\dr},\hat\gamma_\dr), $ 
then an estimator of $\mu$ motivated from the influence function  in Theorem \ref{eif1} is 
\begin{eqnarray}
\hat\mu_\dr&=&\hat E\left[\left\{ \frac{R_2-R_1}{\hat\pi_2^2} - \frac{1-\hat\pi_1}{\hat\pi_1} \frac{R_1}{\hat{\pi}_2} +  \frac{R_1}{\hat{\pi}_1} \right\}  Y\right] \nonumber\\ 
&&-  \hat E\left[\left\{ \frac{R_2-R_1}{\hat\pi_2^2} - \frac{1-\hat\pi_1}{\hat\pi_1} \frac{R_1}{\hat{\pi}_2} +  \frac{R_1}{\hat{\pi}_1}  -1\right\} \frac{ E( Y/\hat\pi_2\mid X,R_2=0;\hat\beta,\hat\gamma_\dr) }{E( 1/ \hat\pi_2\mid X,R_2=0;\hat\beta,\hat\gamma_\dr)} \right].  \label{esti:dr1}
\end{eqnarray}
Equations \eqref{dr1.1}--\eqref{dr1.3} for estimating $(\beta,\alpha_1,\alpha_2)$ remain the same 
as \eqref{reg1.1}, \eqref{ipw1.1}, and   \eqref{ipw1.2}  in the IPW and REG estimation, respectively;
equation \eqref{dr1.4} is a doubly robust estimating equation for $\gamma$.
We summarize the double robustness of the estimators.
\begin{theorem}\label{thm:dr}
Under Assumption \ref{assump:odds} and   regularity conditions described by \citet[][Theorems 2.6 and 3.4]{newey1994large},   $(\hat\alpha_{1,\dr},\hat\gamma_\dr,\hat\mu_\dr)$ are  consistent and asymptotically normal provided  one of the following conditions holds:
\begin{itemize}
\item  $A_1(X;\alpha_1),\Gamma(X,Y;\gamma)$, and $A_2(X;\alpha_2)$ are correctly specified; or   
\item  $A_1(X;\alpha_1),\Gamma(X,Y;\gamma)$, and $\pr_2(Y\mid X;\beta)$ are correctly specified.
\end{itemize}
Furthermore, $\hat\mu_\dr$ attains the semiparametric efficiency bound for the nonparametric model $\M_\npr$
at the intersection model $\M_\ipw\cap \M_\reg$ where all models $\{A_1(X;\alpha_1),A_2(X;\alpha_2), \Gamma(X,Y;\gamma),\pr_2(Y\mid X;\beta)\}$ are correct.
\end{theorem}
This theorem states that $(\hat\alpha_{1,\dr},\hat\gamma_\dr,\hat\mu_\dr)$ are doubly robust against misspecification of the second-call baseline propensity score $A_2(X;\alpha_2)$ and the second-call   outcome distribution  $\pr_2(Y\mid X;\beta)$, provided that the first-call propensity score $\pi_1(X,Y;\alpha_1,\gamma)$ (i.e., $A_1(X;\alpha_1)$ and $\Gamma(X,Y;\gamma)$)
is correctly specified.
Moreover, $\hat\mu_\dr$ is locally efficient if all working models are correct, regardless of the efficiency of  the nuisance estimators.
The outcome mean estimator $\hat\mu_\dr$ has an analogous form to the conventional augmented inverse probability weighted (AIPW) estimator: 
the first part  is an IPW estimator and the second part is an augmentation term involving the outcome regression model.
Compared to the   IPW and REG estimators in the previous section, 
the doubly robust estimator offers one more chance to  correct the bias due to model misspecification.  
However, the DR estimator  is not anticipated to have a smaller variance than the IPW or the REG estimator, 
because  $\M_\npr$ is a larger model with a semiparametric efficiency bound no smaller than that of $\M_\ipw$ or $\M_\reg$.
Besides, if   both $A_2(X;\alpha_2)$ and $\pr_2(Y\mid X;\beta)$  are incorrect, 
the proposed doubly robust estimator will generally also be biased.
The odds ratio $\Gamma(X,Y;\gamma)$ is essential for all three  proposed estimation methods,
because it encodes the degree to which the outcome and the missingness process are correlated.  
In order to estimate  the outcome mean, one must   be able to account for this correlation. 
The first-call baseline propensity score  $A_1(X;\alpha_1)$  needs to be correct for estimation of the odds ratio, and as a result,
$\hat\alpha_{1,\dr}$ is also doubly robust.
If $\pi_1(\alpha_1,\gamma)$ is incorrect, $\hat\mu_\dr$ is in general not consistent even if both $A_2(X;\alpha_2)$ and 
$\pr_2(Y\mid X;\beta)$ are correct---because the latter two models only concern   the second call.
Variance estimation and confidence intervals for the doubly robust approach  also  follow from the general theory for estimation equations \citep[see e.g.,][]{newey1994large},
which can be constructed  based on  the normal approximation or bootstrap under standard regularity conditions.

Doubly robust    methods  have been advocated  in recent years    for   missing data analysis, causal inference, and other problems with  data coarsening. 
Previous proposals  have assumed that   the odds ratio function $\Gamma(X,Y)$ is either known exactly  with  the special case of MAR  \citep[e.g.,][]{scharfstein1999adjusting,lipsitz1999weighted,tan2006distributional,vanderlaan2006targeted,vansteelandt2007estimation,okui2012doubly},
or can be estimated with the aid of an extra instrumental or shadow variable \citep[e.g.,][]{miao2016varieties,sun2018semiparametric,ogburn2015doubly}.
We offer an alternative  approach that achieves doubly robust estimation of both the odds ratio model and the outcome mean with the aid of callback data.
The double robustness of the DR estimator  in Theorem \ref{thm:dr} requires $A_1(X;\alpha_1)$ and $\Gamma(X,Y;\gamma)$ to be correctly specified. 
In Section S3 of the supplement, we also derive the closed form of the efficient influence function in  this semiparametric model 
\[\M_\pa=\{f(X,Y,R_1,R_2):\text{ Assumption \ref{assump:odds} holds}; \text{ $A_1(X;\alpha_1)$ and $\Gamma(X,Y;\gamma)$ are correct}\}.\]
Although  the efficient influence function  in $\M_\pa$ is anticipated to be more efficient than that  in the nonparametric model $\M_\npr$, 
the validity of  this efficient influence function  relies on the correctness of parametric working models for $A_1(X)$ and $\Gamma(X,Y)$.
Instead,  it is preferable to  use    flexible or nonparametric working models to estimate all the nuisance functions and then plug in them into the efficient influence function $\IF(O;\mu)$ in the nonparametric model $\M_\npr$ to estimate $\mu$. 
The following theorem establishes that such an estimator can  deliver  valid $n^{1/2}$  inference for  the functional of interest even if the   nuisance estimators could have   convergence rates considerably slower than $n^{1/2}$. 
Let $\hat A_1$, $\hat A_2$, $\hat \Gamma$ and $\hat f_2$ denote the nuisance estimators using some flexible estimation methods,
with the following norms that measure the convergence rates,
$||\hat A_1 - A_1||_2  =  [ \int \{\hat A_1(x) -A_1(x)\}^2 f(x) dx ]^{1/2}$,
$||\hat A_2 - A_2||_2  =  [ \int \{\hat A_2(x) -A_2(x)\}^2 f(x) dx ]^{1/2}$,
$||\hat\Gamma - \Gamma||_2 = [ \int \int \{\hat\Gamma(x,y) - \Gamma(x,y)\}^2f_2(y\mid x) dy f(x)dx ]^{1/2}$,
$||\hat f_2 - f_2||  =  [ \int \{ \int |\hat f_2(y\mid x) - f_2(y\mid x)|dy \}^2 f(x) dx ]^{1/2}.$

\begin{theorem} \label{thm:dr:fl}
	Under Assumption \ref{assump:odds}, suppose   $c \leq \pi_1,\pi_2 \leq 1-c$ for some constant $c>0$,     $|Y| \leq M$ for some constant $M>0$, and the   nuisance estimators satisfy that
(i) $c \leq \hat\pi_1,\hat \pi_2 \leq 1-c$;
(ii) $||\hat{A}_1 - A_1||_2 = o_p(1), ||\hat{A}_2 - A_2||_2 = o_p(1), ||\hat{\Gamma} - \Gamma||_2 = o_p(1), ||\hat{f}_2 - f_2||  = o_p(1)$;
and (iii) $\hat A_1$, $\hat A_2$, $\hat \Gamma$, $\hat f_2$ and $A_1$, $A_2$, $\Gamma$, $f_2$ are in a Donsker class.
 Letting
\begin{eqnarray}
	\hat\mu_{\dr2}&=&\hat E\left[\left\{ \frac{R_2-R_1}{\hat\pi_2^2} - \frac{1-\hat\pi_1}{\hat\pi_1} \frac{R_1}{\hat{\pi}_2} +  \frac{R_1}{\hat{\pi}_1} \right\}  Y\right] \nonumber\\ 
	&&-  \hat E\left[\left\{ \frac{R_2-R_1}{\hat\pi_2^2} - \frac{1-\hat\pi_1}{\hat\pi_1} \frac{R_1}{\hat{\pi}_2} +  \frac{R_1}{\hat{\pi}_1}  -1\right\} \frac{ \int y e^{-\hat \Gamma}/\hat\pi_2 \hat f_2 dy }{\int e^{-\hat \Gamma}/ \hat\pi_2 \hat f_2 dy} \right],  \label{esti:np}
\end{eqnarray}
 then we have
	\begin{eqnarray*}
		\hat\mu_{\dr2} - \mu &=& \hat E\{\IF(O;\mu)\} + O_p({\rm Rem}) + o_p(n^{-1/2}),
	\end{eqnarray*}
	where
\begin{eqnarray} \label{Rem}
	{\rm Rem} &= &  ||\hat{A}_1 - A_1||_2^2 + ||\hat{\Gamma} - \Gamma||_2^2 + ||\hat{A}_1 - A_1||_2\cdot ||\hat{\Gamma} - \Gamma||_2 + ||\hat{A}_2 - A_2||_2 \cdot ||\hat{f}_2 - f_2||  \nonumber\\
	&& + (||\hat{A}_1 - A_1||_2 + ||\hat{\Gamma} - \Gamma||_2) \cdot ( ||\hat{A}_2 - A_2 ||_2 + ||\hat{f}_2 -f_2||).
\end{eqnarray}
\end{theorem}

Theorem \ref{thm:dr:fl} generalizes the parametric doubly robust estimation.  
Theorem \ref{thm:dr:fl} implies that $\hat\mu_{\dr2}$ has influence function $\IF(O;\mu)$ if the bias term ${\rm Rem}$, a function of estimation errors of nuisance estimators, is of order $o_p(n^{-1/2})$. 
The bias of $\hat\mu_{\dr2}$ is of second order, i.e., 
it only depends on   squared errors of nuisance estimators. 
This result opens the way to implement more sophisticated estimation methods  built on  flexible,  data-adaptive machine learning or  nonparametric    nuisance models   to achieve $n^{1/2}$-consistent estimation of the functional of interest,  
provided that    the  nuisance estimators  have  estimation error of order smaller than $n^{-1/4}$. 
For example, one can use sieve estimation with polynomials as basis functions to approximate nuisance functions.
Suppose   $A_1(X)$ and $A_2(X)$ are unrestricted, $\Gamma(X,Y) = G(X)Y$, $f_2(Y\mid X) \sim N(B_1(X),\exp\{B_2(X)\})$ for some unknown functions $G,B_1,B_2$ of $X$,
and let $\{\psi_j(X)\}_{j=1}^\infty$ denote a sequence of polynomials of $X$ and $V_{i,n}(X) = \{\psi_1(X),\dots,\psi_{J_{i,n}}(X)\}^\t$ for some $J_{i,n}$, $i=1,\dots,5$,
then we  approximate the nuisance functions by $A_{1,n}(X) = \alpha_{1,n}^\t V_{1,n}(X)$, $A_{2,n}(X) = \alpha_{2,n}^\t V_{2,n}(X)$, $G_{n}(X) = \gamma_{n}^\t V_{3,n}(X)$, $B_{1,n}(X) = \beta_{1,n}^\t V_{4,n}(X)$, $B_{2,n}(X) = \beta_{2,n}^\t V_{5,n}(X)$.
The nuisance estimators are obtained with $\alpha_{1,n}$, $\alpha_{2,n}$, $\gamma_{n}$, $\beta_{1,n}$, $\beta_{2,n}$   estimated by solving equations   \eqref{dr1.1}-\eqref{dr1.4}.
Finally, $\hat\mu_{\dr2}$ is obtained by plugging in the nuisance estimators into \eqref{esti:np}.
Theorem \ref{thm:dr:fl} imposes the Donsker condition \citep{van1996weak}  on  the complexity of the nuisance models.
It could be relaxed to admit more complicated working models by  employing the cross-fitting method \citep[e.g.,][]{robins2008higher,chernozhukov2018double}.

Estimation methods with second-order bias  have been   extensively studied by \citet{benkeser2017doubly,kennedy2017non,athey2018approximate,tan2020model,rotnitzky2021characterization,dukes2021inference,chernozhukov2018double},  
and Theorem \ref{thm:dr:fl} contributes to this literature with an instance in  nonignorable missingness data analysis.
Theorem \ref{thm:dr:fl} reveals    the mixed bias property \citep{rotnitzky2021characterization} in  the nonignorable missing data setting when the odds ratio function needs to be estimated: if   $A_1$ and $\Gamma$ are estimated parametrically, the bias of  $\hat\mu_{\dr2}$  only  depends on $||\hat{A}_2 - A_2||_2$ and $||\hat{f}_2 - f_2||$ through their product. 
The justification of this result involves a more sophisticated analysis of the bias term instead of directly applying the second-order Taylor expansion technique,
and our technical contribution may shed some light on characterizing the bias terms of estimators with multiple nuisance components.

\section{Extensions}

\subsection{Estimation of a general full-data functional}
We extend the  proposed IPW, REG, and DR methods to estimation of a general smooth full-data functional   $\theta$, defined as the unique solution to a given estimating equation $E\{m(X,Y;\theta)\}=0$.
Familiar examples include the outcome mean   with $m(X,Y;\mu)=Y-\mu$;
the least squares coefficient  with $m(X,Y;\theta)= X(Y-\theta^\t X)$;
and the instrumental variable estimand in causal inference with $m(X,Y;\theta)= Z(Y-\theta^\t W)$, 
where $W$ and $Z$ are subvectors of $X$, $\theta$ is the causal effect of $W$ on $Y$ and is subject to unmeasured confounding and $Z$ is a set of  instrumental variables used for confounding adjustment in causal inference.
Given  estimators of the nuisance parameters, the IPW and REG estimators of $\theta$ can be obtained by solving 
\[0=\hat E\left\{\frac{R_2}{\hat p_2} m(X,Y;\theta)\right\},\]
\[0=\hat E [R_2 m(X,Y;\theta) + (1-R_2)  E\{m(X,Y;\theta) \mid X,R_2=0;\hat\beta,\hat\gamma_\reg\}],\]
respectively.
Letting $m(\theta)\equiv m(X,Y;\theta)$, the efficient influence function for $\theta$ in     model $\M_\npr$ is 
\begin{eqnarray}\label{geif1}
&&\hspace{-1cm}\IF(O;\theta)= -\left[\frac{\partial E\{m(\theta)\}}{\partial \theta}\right]^{-1}\nonumber\\
&&\quad\cdot \left[\left\{\frac{R_2-R_1}{\pi^2_2} - \frac{1-\pi_1}{\pi_1} \frac{R_1}{\pi_2}+ \frac{R_1}{\pi_1} - 1\right\}
\left\{m(\theta) - \frac{E(m(\theta)/\pi_2  \mid X,R_2=0)}{E(1 /\pi_2\mid X,R_2=0)}\right\} + m(\theta) \right].
\end{eqnarray}
We prove this result in the supplement.
Then a  doubly robust and locally efficient estimator of $\theta$ can be constructed by solving $\hat E\{\IF(O;\theta)\}=0$, 
after obtaining  nuisance estimators.

\subsection{Estimation with multiple callbacks}
When multiple calls $(K\geq 3)$ are available, identification of $f(X,Y,R_1,\dots,R_K)$   equally holds  under Assumption \ref{assump:odds}, even if the propensity scores and outcome distributions  for the third and later call attempts are completely unrestricted.
Besides, identification can be achieved if Assumption \ref{assump:odds} holds for at least two known adjacent calls.
To see this, suppose that the stableness of resistance  holds for the $k$th and $k+1$th calls. 
By viewing the $k$th and $k+1$th calls as the first two calls in a subsurvey on nonrespondents from the first $k-1$th calls (i.e., $R_{k-1}=0$) and by applying the proposed approach, 
we can identify $f(X,Y,R_k,R_{k+1}\mid R_{k-1}=0)$.
Noting that $f(X,Y,R_k,R_{k+1}\mid R_{k-1}=1)$ is available from the observed data, 
we can  identify $f(X,Y,R_{k-1},R_k,R_{k+1})$ and  thus $f(X,Y)$.
Identifying $f(X,Y)$ suffices to  identify $f(X,Y,R_1,\ldots, R_K)$.
In this case, the proposed IPW, REG, and DR estimation methods developed for   the setting with two calls 
still work but  they  are agnostic to the  data obtained in later calls.
In Section S4 of the supplement, we extend the proposed IPW, REG, and DR estimators to the setting with multiple callbacks, which can incorporate all observations on $(X,Y)$.

\section{Simulation}    
We evaluate the performance of the proposed estimators and assess their robustness against  misspecification of working models and violation of the key identifying assumption   via simulations.   
We conduct simulations for  both continuous and binary outcome settings.
Simulation results for these two settings are analogous, 
and to save space, here we only present simulations for  the continuous setting and defer   the binary  setting   to Section S6.3 of the supplement. 

Let $X=(1,X_1,X_2)^\t$ and  $\widetilde X=(1,X_1^2,X_2^2)$ with $X_1,X_2$ independent following a uniform distribution ${\rm Unif}(-1,1)$.
We consider four   data generating scenarios according to different  choices for 
the second-call baseline propensity score and the second-call outcome distribution.
The following Table  \ref{tbl:dgm1} presents the data generating mechanisms and the working models for estimation.
\vspace{-1em}
\begin{table}[H]
\centering
\caption{Data generating model and working model for estimation}\label{tbl:dgm1}
\begin{tabular}{ccccc}
&\multicolumn{4}{c}{Data generating model}\\
&\multicolumn{4}{c}{$\pi_1=\expit(\alpha_1^\t X + \gamma Y),\pi_2=\expit(\alpha_2^\t  W_1 + \gamma Y), f_2(Y\mid X)\thicksim N(\beta^\t W_2,  \sigma^2)$}\\
&\multicolumn{4}{c}{Four  scenarios with different choices of $(W_1,W_2)$}\\
& TT&FT&TF&FF\\
&$W_1=X,W_2=X$&$W_1=\widetilde X,W_2=X$&$W_1=X,W_2=\widetilde X$&$W_1=\widetilde X,W_2=\widetilde X$\\
$\alpha_1^\t $& (0, 0.6, 0.5)& (-0.35, -0.5, 0.7)  & (-1, 1, -0.1)& (-0.3, -0.5, 1)  \\
$\alpha_2^\t$&  (1.4, -0.5, 0.2)& (-0.5, 1.8, 1) & (0.5, 1, -0.1)  & (-0.4, 0.8, 0) \\
$\beta^\t$&  (0.6, 1.0, 0.3) & (-0.8, 5, 3.5)  & (-0.5, 5, -1) & (-1.5, 4, 3)  \\
$\gamma$&  0.13 & 0.12 & 0.5 & 0.25 \\ 
$\sigma$& 3 & 2 & 0.4 & 0.25 \\
&\multicolumn{4}{c}{Working model for estimation}\\
&\multicolumn{4}{c}{$\pi_1=\expit(\alpha_1^\t X + \gamma Y),\pi_2=\expit(\alpha_2^\t X + \gamma Y), f_2(Y\mid X)\thicksim N(\beta^\t X,  \sigma^2)$}
\end{tabular}
\end{table}
Hence, the working model for the  second-call baseline propensity score  is correct in Scenarios (TT) and (TF),
the second-call outcome model is correct in Scenarios (TT) and (FT),  and they both are  incorrect in Scenario (FF).
The first-call baseline propensity score and the odds ratio models are correct in all scenarios.
For estimation,  we implement the proposed IPW, REG, and  DR  methods $(\hat{\mu}_{\ipw}$,$\hat{\mu}_{\reg},\hat{\mu}_{\dr})$ based on parametric working models to estimate the outcome mean and the odds ratio parameter.
We  compute the variance of these estimators  and then construct  the confidence interval based on the normal approximation and evaluate the  coverage rate.
For comparison, we   also implement  standard   IPW$_\mar$, REG$_\mar$, DR$_\mar$ estimators ($\hat{\mu}^{\rm{mar}}_{\ipw}$,$\hat{\mu}^{\rm{mar}}_{\reg},\hat{\mu}^{\rm{mar}}_{\dr}$) that are based on MAR,  with the number of callbacks included as an additional covariate.

For each scenario, we replicate 1000 simulations at sample size 5000.
Figures \ref{fig:simu1mu} and \ref{fig:simu1gamma} show the  bias for   estimators of  the outcome mean and the odds ratio parameter, respectively, and Table \ref{tbl:coverage} shows the coverage rates for the corresponding confidence intervals.
In Scenario (TT) where all working models are correct, all three proposed estimators have little bias and   the 95\% confidence intervals  have  coverage rates    close to 0.95;
in Scenario (FT) the second-call baseline propensity score model is incorrect but the second-call outcome model is correct, then the IPW estimator has large bias with    coverage rate well below 0.95 and the REG estimator has little bias with the coverage rate close to 0.95;
conversely, in Scenario (TF) the second-call baseline propensity score model is correct but the second-call outcome model is incorrect, then the IPW estimator has little bias with an appropriate coverage rate and the REG estimator has large bias with an  undersized coverage rate.
The DR estimator has  little bias with    appropriate coverage rates in all the three  scenarios when at least one of  these two working models is correct.
In Scenario (FF) where both working models are incorrect, all three proposed  estimators are biased.
These numerical results   show the robustness of  $\hat{\mu}_{\dr}$ against partial  misspecification of working models. 
We  compute the standard deviations  of the three proposed estimators in Table \ref{tbl:sd}  in  the supplement, which are  close  when all working models are   correct.
However, the three standard MAR estimators have large bias in all four scenarios.
Therefore, we recommend  the proposed    methods  for nonresponse adjustment with callback data, 
and we suggest  to use different working models to gain more robust inferences.

In Section S6.2 of the supplement, we conduct additional simulations to compare the performance of $\hat \mu_{\dr2}$ using   flexible  nuisance models and the IPW, REG and DR   estimators that are based on parametric working models.
We consider six data generating scenarios,   
including  two scenarios where all parametric working models are correct or incorrect, respectively,  and four scenarios where each of $A_1,A_2,f_2,\Gamma$ is incorrect, respectively.
The   estimator $\hat\mu_{\dr2}$ has little bias with appropriate coverage rates in all   six scenarios.
But the IPW, REG and DR estimators could have large bias if the required working models are incorrect;
in particular, these three estimators are biased if the odds ratio model is incorrectly specified.
The numerical results show the robustness of $\hat\mu_{\dr2}$, which delivers  valid inference in scenarios where the nuisance models are complicated and   parametric working models  are more likely to be misspecified.

In addition to model misspecification, we also  evaluate sensitivity  of the inference against violation of the identifying assumption.
We include    a sensitivity analysis  to evaluate the performance of the proposed estimators    when 
call-specific odds ratios arise, i.e.,  the stableness of resistance assumption is not met.
The difference between the odds ratios is used as the  sensitivity parameter to  capture the degree to which the    stableness of resistance assumption is violated.
The proposed estimators exhibit  small  bias when the difference of odds ratios varies within a moderate range but it could become severe if there were a big gap between  the odds ratios.
To obtain  reliable  inferences in practice, we recommend such sensitivity analysis for assessing  robustness of inference against violation of the stableness of resistance assumption.

\begin{figure}[H]
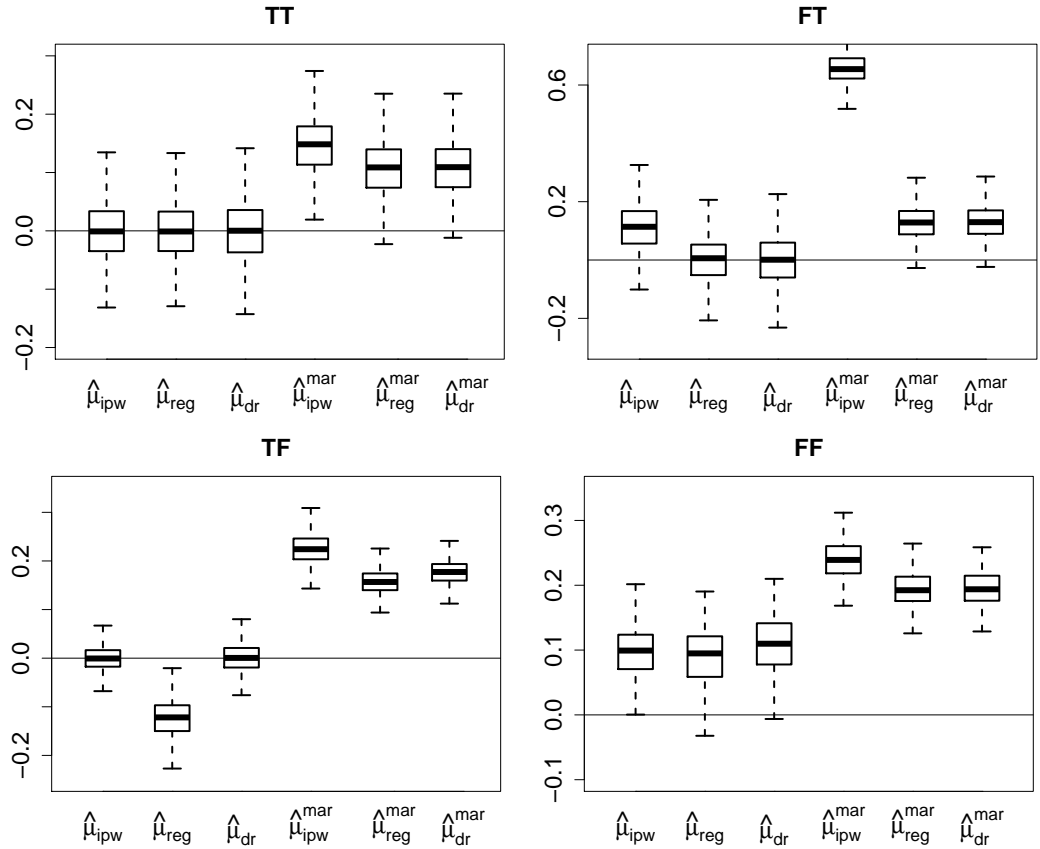

\graphicspath{{Rscripts/Simulations/Simulations_Continuous/Results/}}
\centering
\includegraphics[scale=0.5]{TT/bias_mu.pdf}
\includegraphics[scale=0.5]{FT/bias_mu.pdf}\\
\includegraphics[scale=0.5]{TF/bias_mu.pdf}
\includegraphics[scale=0.5]{FF/bias_mu.pdf}
\caption{Bias for estimators of $\mu$ in the continuous outcome setting.} \label{fig:simu1mu}
\end{figure}
\begin{figure}[H]
\graphicspath{{Rscripts/Simulations/Simulations_Continuous/Results/}}
\includegraphics[scale=0.5]{TT/bias_gamma.pdf}
\includegraphics[scale=0.5]{FT/bias_gamma.pdf}
\includegraphics[scale=0.5]{TF/bias_gamma.pdf}
\includegraphics[scale=0.5]{FF/bias_gamma.pdf}
\caption{Bias for estimators of $\gamma$  in the continuous outcome setting.}\label{fig:simu1gamma}
\end{figure}

\begin{table}[H]
\centering
\caption{Coverage rate of $95\%$ confidence interval in the continuous outcome setting} \label{tbl:coverage} 
\begin{tabular}{ccccccccccc}
&\multicolumn{6}{c}{$\mu$} && \multicolumn{3}{c}{$\gamma$}\\ 
Scenarios& IPW & REG  & DR &IPW$_\mar$ &REG$_\mar$&DR$_\mar$& &IPW  & REG  & DR   \\
TT &{0.955}  & {0.955} & {0.950} &{0.109}  & {0.362} & {0.363}& &{0.955} & {0.955} & {0.956} \\ 
FT & {0.629} & {0.948} & {0.952} &{0.000}  & {0.442} & {0.426}& &{0.322} &  {0.957}  & {0.945} \\
TF & {0.944}  & {0.125} & {0.952} &{0.000}  & {0.000} & {0.000}& &{0.947} & {0.501} & {0.944}    \\
FF & {0.286} & {0.446} & {0.287}  &{0.000}  & {0.000} & {0.000}& &{0.489} & {0.614}  & {0.675}  
\end{tabular} 
\end{table}

\section{Real data application}

The Consumer Expenditure Survey \citep[CES;][]{ces2013measuring}  is a nationwide survey conducted by the U.S. Bureau of Labor Statistics to  find out  how American households make and spend  money. 
It comprises    two  surveys: the Quarterly Interview  Survey  on large and recurring expenditures  such as rent and utilities, and the Diary Survey on small and high frequency purchases, such as food and clothing.
The survey  data are  released  annually  since 1980, which  contain detailed  callback history.           
We analyze   the public-use microdata from  the Quarterly Interview  Survey in the fourth quarter of 2018,
available from \href{https://www.bls.gov/cex/pumd\_data.htm\#csv}{\texttt{https://www.bls.gov/cex/pumd\_data.htm\#csv}}.
The nonresponse  of frame variables is concurrent due to contact failure or refusal
and no fully-observed baseline covariates are available in this dataset.
For illustration, we analyze this dataset to study the expenditures  on housing and on utilities, fuels and public services. 	  
This survey contains   9986 households and  277 of them   with  extremely large or small expenditures   are removed in our analysis.
The maximum number of contact attempts the interviewers made in this survey is about 30. 
Figure \ref{fig:ce1} (a) shows the    cumulative response rate.
The overall response rate is about 0.6 and  80\% of the respondents completed the survey within the first five contact attempts.  
Let $(Y_1,Y_2)$ denote the logarithm of the  expenditure  on housing  and   on utilities, fuels and public services, respectively.
Figure \ref{fig:ce1} (b) shows  the   mean of  $(Y_1,Y_2)$ for respondents to each call attempt. 
The mean of $Y_1$ for respondents  gradually increases  with the number of contacts, 
and the mean of  $Y_2$      has a sharp increase in the second and third contacts and fluctuates in later  contacts.
These results suggest that the delayed  respondents are  likely to have higher expenditures and  the nonresponse is likely dependent on the expenditures,  in particular, on $Y_1$.

\begin{figure}[H]
\graphicspath{{Rscripts/}}
\centering
\subfigure[Cumulative response rate]{
\includegraphics[scale=0.45]{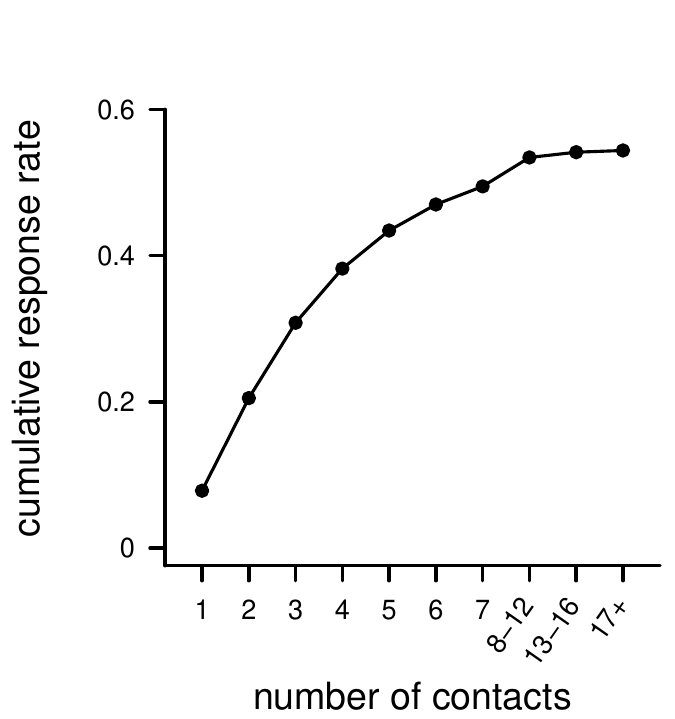}}
\hfill
\subfigure[Outcomes mean]{\includegraphics[scale=0.45]{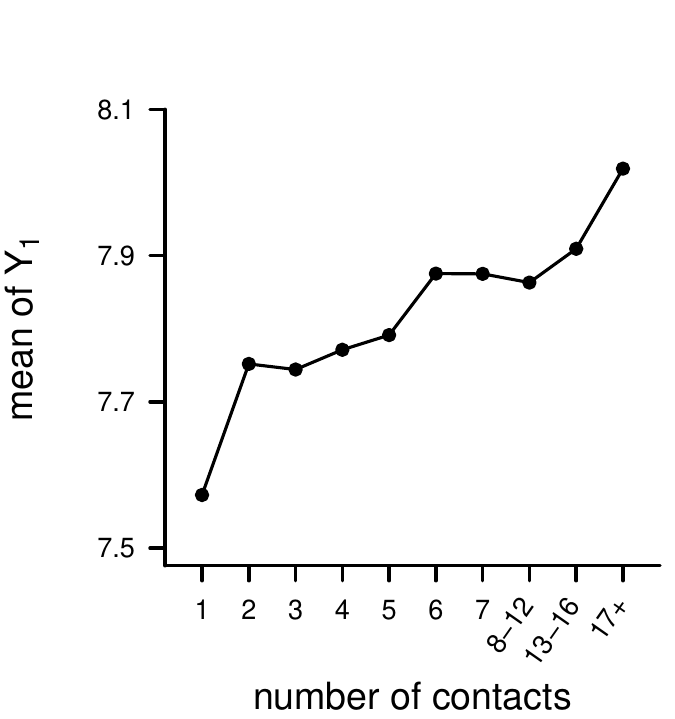}
\includegraphics[scale=0.45]{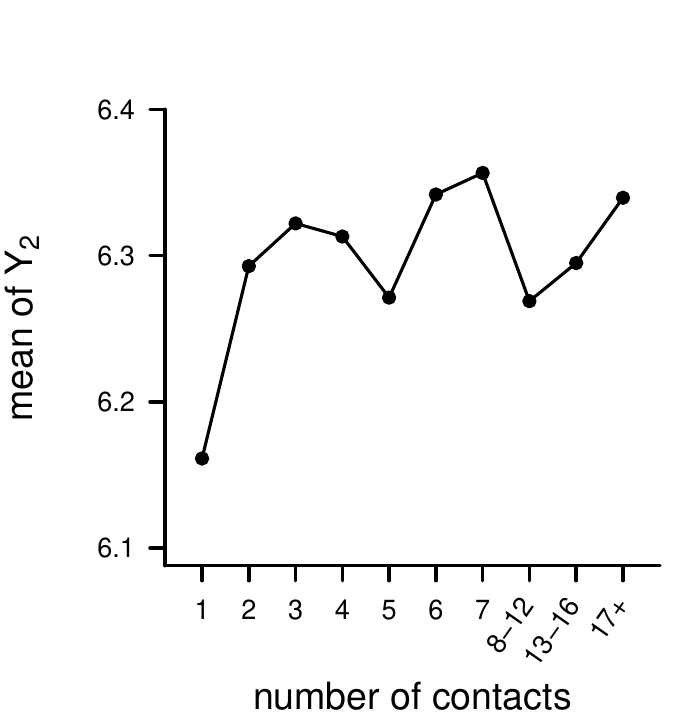}}
\caption{Response rates and outcomes mean for  respondents  in the CES application.}\label{fig:ce1}
\end{figure}

We apply the proposed   methods to estimate the    outcomes mean  $\mu=(\mu_1,\mu_2)=(E(Y_1),E(Y_2))$.
Analogous to previous survey studies \citep[e.g.,][]{qin2014semiparametric,boniface2017assessment}, 
we split the   contact attempts into two stages: 1--2 calls (early contact) and 3+ calls (late contact).  
Among the 9709 households we analyze, 1992 responded in the first  stage,
3287 responded later, and   4430 never responded.
We       fit the data  with working models  $\pi_1=\expit(\alpha_1 +\gamma_1 Y_1+\gamma_2 Y_2),\pi_2=\expit(\alpha_2 +\gamma_1 Y_1+\gamma_2 Y_2)$ and  $\pr_2(Y)\thicksim N(\mu,\Sigma)$, and apply the proposed IPW, REG, and DR methods to  estimate the outcomes mean.
The common odds ratio parameters $(\gamma_1,\gamma_2)$ reveal  that the resistance to respond caused by the outcomes  remain the same in these two  contact stages.
We also compute the  complete-case (CC) sample mean and apply   standard  inverse probability weighting, regression and doubly robust estimation methods that are based on MAR and include  the number of contacts   as a covariate, which are respectively denoted by IPW$_\mar$, REG$_\mar$, and DR$_\mar$.

Table \ref{tbl:ce} presents the     estimates. 
The DR point estimate of $\mu_1$ is   7.842 with 95\% confidence interval (7.800, 7.884),
and of $\mu_2$ is 6.345  with  (6.296, 6.393).
The CC estimate of $\mu_1$ is 7.756 (7.734, 7.778) and of $\mu_2$ is 6.285 (6.264, 6.306); the DR$_\mar$ estimate of $\mu_1$ is 7.767 (7.744, 7.790) and of $\mu_2$ is 6.287 (6.266, 6.309).
The     IPW and REG  methods produce estimates    close to DR;
however, the CC and standard MAR estimates  of the outcomes mean, in particular of $\mu_1$,   are well below the DR estimate.
As shown in      Figure \ref{fig:ce1} (b), this can be partially explained by the fact   that  the outcomes mean in early respondents  is  lower    than  in the delayed,    
and therefore,  the outcomes mean in  the  respondents  is   likely   lower than in the  nonrespondents.
The estimation results of the odds ratio parameters reinforce this conjecture:
the  IPW, REG, and DR estimates of the odds ratio parameters are all negative, 
suggesting that high-spending people are more reluctant to respond or more difficult to contact. 
The odds ratio estimate for expenditure on housing is statistically significant at  level 0.01,  
although it is not significant for  expenditure on utilities, fuels and public services. 
This is evidence for missingness not at random.
These results  indicate that   the expenditure on housing  play a more important role in  the response process;
this may be because the survey takes personal home visit   as one of the main modes  of interview and people 
with high expenditure on housing are more difficult to reach. 
Increasing the variety  of     interview modes may potentially alleviate such  nonignorable nonresponse.

\begin{table}[H]
\centering
\caption{Point estimate,   95\% confidence interval, and p-value for the CES application}  \label{tbl:ce}

\setlength\tabcolsep{4pt}
\begin{tabular}{lccccccccclcccc}
& \multicolumn{7}{c}{$\mu_1$} 
&& \multicolumn{3}{c}{$\gamma_1$}&\\ 
&IPW & REG & DR &  IPW$_\mar$ & REG$_\mar$ & DR$_\mar$  & CC 
&& IPW & REG & DR\\
Estimate &7.859 &7.861 &7.842 & 7.767& 7.769& 7.767 &7.756  
&& -0.269&  -0.258 & -0.238   \\ [6pt]
\makecell[l]{CI or\\ p-value}& \makecell[c]{(7.785,\\ 7.932)} &\makecell[c]{(7.810,\\ 7.912)}&\makecell[c]{(7.800,\\ 7.884)} & \makecell[c]{(7.744,\\ 7.789)}
&\makecell[c]{(7.746,\\ 7.791)} &
\makecell[c]{(7.744,\\ 7.790)}
& \makecell[c]{(7.734,\\ 7.778)}
&&  0.004& 0.001 & 0\\ [10pt]
& \multicolumn{7}{c}{$\mu_2$} 
&& \multicolumn{3}{c}{$\gamma_2$}&\\ 
&IPW & REG & DR & IPW$_\mar$ & REG$_\mar$ & DR$_\mar$  & CC 
&& IPW & REG & DR\\
Estimate &6.346 &6.350 &6.345 &6.287  & 6.288 & 6.287 &6.285  
&& -0.028 &  -0.042 & -0.056   \\  [6pt]
\makecell[l]{CI or\\ p-value}&\makecell[c]{(6.269,\\ 6.423)}&\makecell[c]{(6.297,\\ 6.403)}&\makecell[c]{(6.296,\\ 6.393)}& \makecell[c]{(6.266 ,\\ 6.309)}
&\makecell[c]{(6.267,\\ 6.310)} &
\makecell[c]{(6.266,\\ 6.309)} &\makecell[c]{(6.264,\\ 6.306)}
&&  0.812  & 0.658 & 0.521            
\end{tabular}                         	
\end{table}

We   further conduct   sensitivity analysis  to  assess robustness of the above results      against   violation of the stableness of resistance assumption.
The sensitivity analysis shows that our    results are not sensitive to mild violations of the stableness of resistance assumption.
In particular, when the sensitivity parameter varies within a moderate range, 
the DR estimate of $\mu_1$ remains  larger than the complete-case (CC) sample mean and the estimates based on MAR, 
and the estimate of $\gamma_1$ remains significantly negative.
Such results  reinforce our finding that high-spending people are more reluctant to respond or more difficult to contact.
Details of the sensitivity analysis  are relegated to Section   S7 of  the supplement.

\section{Discussion}

We establish  a novel framework for nonresponse adjustment with callback data, which further illustrates the usefulness and extends the application of callback data.
Although not all surveys provide callback  data, their availability is increasing in modern surveys.
The stableness of resistance assumption is key to our framework.
Under this assumption, we establish nonparametric identification and propose a suite of novel estimators including a doubly robust one,
which  extend  previous parametric approaches and  further elucidate the underlying source for nonresponse adjustment with callback data.
We caution that the stableness of resistance assumption is untestable based on observed data.
Therefore, analogous to various missing data problems  \citep{molenberghs2008every,miao2016varieties,sun2018semiparametric}, 
its  validity should be justified based on domain-specific knowledge and  needs to be investigated on a case-by-case basis. 
We have clarified the motivation, implication and limitation of the stableness of resistance assumption to facilitate the justification in practice.
Even if the assumption does not hold, 
our approach constitutes a valid test of whether the missingness is entirely MAR--because the stableness of resistance assumption naturally holds under the null hypothesis of MAR.
Besides,   sensitivity analysis  is   warranted  to assess robustness of inference against violation of  the assumption.

We have employed an odds ratio parametrization and adopted practical parametric working models in the IPW, REG and DR estimators.
There exist other parametrizations and we describe estimation under an alternative parametrization  in Section S2 of the supplement.
\cite{kang2007demystifying} cautioned for  potentially disastrous bias of certain DR estimators  under MAR  when all parametric working models are incorrect.
However, previous authors have proposed alternative  constructions of nuisance estimators and DR estimators   to alleviate this problem,
see e.g. \cite{tan2010bounded,vermeulen2014biased,tsiatis2011improved} and the discussions alongside \cite{kang2007demystifying}.
Besides,  multiply robust  estimation in the sense of \cite{vansteelandt2007estimation} is also of interest.
In addition, our nonparametric identification  and estimation results open the way to more sophisticated estimation methods built on complex  models, 
such as series or sieve estimation.
In particular, our proposal in Theorem \ref{thm:dr:fl} allows for flexible estimation of the odds ratio function, without requiring it to be known or follow a parametric model,
which extends the methods for nonignorable missing data analysis.
For large to high-dimensional covariates,  a heuristic approach for variable selection   is to  include  penalties (e.g. LASSO)  into  the optimization of estimating equations for the nuisance parameters \citep[e.g.][]{garcia2010variable,fang2016model}. 
However, there remain   challenges to the variable selection in the presence of nonignorable missing data and callbacks.
It is of interest to incorporate these approaches to improve  the proposed estimation methods.

We   considered  a single action response process---a call attempt either succeeds or fails;
however,  there may exist several  dispositions,  e.g.,  interview, refusal, other non-response or final non-contact \citep{biemer2013using}.
Concurrent nonresponse  is often the case when the missingness is due to failure of contact,  
but in practice   different  frame variables may be  observed in different call attempts. 
In this case one may combine the callback design and the graphical model  \citep[e.g.,][]{sadinle2017itemwise,malinsky2020semiparametric,mohan2021graphical}   to account for complex patterns of missingness.
In addition to   nonresponse adjustment, 
callback data  are also useful  for  the design and organization of surveys, e.g., allocation of time and staff resources.
The integration of the callback design  and other tools (e.g., instrumental variables) may be useful for   handling simultaneous problems of   nonresponse  and confounding or other deficiencies in survey and observational studies.
It is of interest to pursue these extensions.

\section*{Acknowledgements}
We are grateful for   valuable comments from   the editor, the associate editor  and three anonymous reviewers.
This work is partially supported by National Key R\&D Program of China (2022YFA1008100) and National Natural Science Foundation of China (12292983, 12071015).

\section*{Supplementary material}
Supplementary materials online    include further illustration for the stableness of resistance assumption,
estimation under an alternative parametrization,
the efficient influence function under the semiparametric model that $A_1(X;\alpha_1),\Gamma(X,Y;\gamma)$ are correctly specified,
estimation with multiple callbacks,
proof of theorems, propositions, and important equations,
additional  simulations and real data analysis results,
and codes and data for reproducing the simulations and application.


\bibliographystyle{chicago}
\bibliography{CausalMissing}

\begin{thebibliography}{}

\bibitem[\protect\citeauthoryear{Alho}{Alho}{1990}]{alho1990adjusting}
Alho, J.~M. (1990).
\newblock Adjusting for nonresponse bias using logistic regression.
\newblock {\em Biometrika\/}~{\em 77}, 617--624.

\bibitem[\protect\citeauthoryear{Athey, Imbens, and Wager}{Athey
  et~al.}{2018}]{athey2018approximate}
Athey, S., G.~W. Imbens, and S.~Wager (2018).
\newblock Approximate residual balancing: debiased inference of average
  treatment effects in high dimensions.
\newblock {\em Journal of the Royal Statistical Society: Series B\/}~{\em 80},
  597--623.

\bibitem[\protect\citeauthoryear{Benkeser, Carone, Laan, and Gilbert}{Benkeser
  et~al.}{2017}]{benkeser2017doubly}
Benkeser, D., M.~Carone, M.~V.~D. Laan, and P.~Gilbert (2017).
\newblock Doubly robust nonparametric inference on the average treatment
  effect.
\newblock {\em Biometrika\/}~{\em 104}, 863--880.

\bibitem[\protect\citeauthoryear{Biemer, Chen, and Wang}{Biemer
  et~al.}{2013}]{biemer2013using}
Biemer, P.~P., P.~Chen, and K.~Wang (2013).
\newblock Using level-of-effort paradata in non-response adjustments with
  application to field surveys.
\newblock {\em Journal of the Royal Statistical Society: Series A\/}~{\em 176},
  147--168.

\bibitem[\protect\citeauthoryear{Boniface, Scholes, Shelton, and
  Connor}{Boniface et~al.}{2017}]{boniface2017assessment}
Boniface, S., S.~Scholes, N.~Shelton, and J.~Connor (2017).
\newblock Assessment of non-response bias in estimates of alcohol consumption:
  Applying the continuum of resistance model in a general population survey in
  england.
\newblock {\em PloS ONE\/}~{\em 12}, e0170892.

\bibitem[\protect\citeauthoryear{Chen, Li, and Qin}{Chen
  et~al.}{2018}]{chen2018generalization}
Chen, B., P.~Li, and J.~Qin (2018).
\newblock Generalization of \uppercase{H}eckman selection model to nonignorable
  nonresponse using call-back information.
\newblock {\em Statistica Sinica\/}~{\em 28}, 1761--1785.

\bibitem[\protect\citeauthoryear{Chen}{Chen}{2007}]{chen2007semiparametric}
Chen, H.~Y. (2007).
\newblock A semiparametric odds ratio model for measuring association.
\newblock {\em Biometrics\/}~{\em 63}, 413--421.

\bibitem[\protect\citeauthoryear{Chernozhukov, Chetverikov, Demirer, Duflo,
  Hansen, Newey, and Robins}{Chernozhukov
  et~al.}{2018}]{chernozhukov2018double}
Chernozhukov, V., D.~Chetverikov, M.~Demirer, E.~Duflo, C.~Hansen, W.~Newey,
  and J.~Robins (2018).
\newblock Double/debiased machine learning for treatment and structural
  parameters.
\newblock {\em The Econometrics Journal\/}~{\em 21}, C1--C68.

\bibitem[\protect\citeauthoryear{Clarsen, Skogen, Nilsen, and Aar{\o}}{Clarsen
  et~al.}{2021}]{clarsen2021revisiting}
Clarsen, B., J.~C. Skogen, T.~S. Nilsen, and L.~E. Aar{\o} (2021).
\newblock Revisiting the continuum of resistance model in the digital age: a
  comparison of early and delayed respondents to the norwegian counties public
  health survey.
\newblock {\em BMC Public Health\/}~{\em 21}, 730.

\bibitem[\protect\citeauthoryear{Daniels, Jackson, Feng, and White}{Daniels
  et~al.}{2015}]{daniels2015pattern}
Daniels, M.~J., D.~Jackson, W.~Feng, and I.~R. White (2015).
\newblock Pattern mixture models for the analysis of repeated attempt designs.
\newblock {\em Biometrics\/}~{\em 71}, 1160--1167.

\bibitem[\protect\citeauthoryear{Deming}{Deming}{1953}]{deming1953probability}
Deming, W.~E. (1953).
\newblock On a probability mechanism to attain an economic balance between the
  resultant error of response and the bias of nonresponse.
\newblock {\em Journal of the American Statistical Association\/}~{\em
  48\/}(264), 743--772.

\bibitem[\protect\citeauthoryear{D'Haultf\oe{}uille}{D'Haultf\oe{}uille}{2010}]{d2010new}
D'Haultf\oe{}uille, X. (2010).
\newblock A new instrumental method for dealing with endogenous selection.
\newblock {\em Journal of Econometrics\/}~{\em 154}, 1--15.

\bibitem[\protect\citeauthoryear{Drew and Fuller}{Drew and
  Fuller}{1980}]{drew1980modeling}
Drew, J. and W.~A. Fuller (1980).
\newblock Modeling nonresponse in surveys with callbacks.
\newblock In {\em Proceedings of the Section on Survey Research Methods of the
  American Statistical Association}, pp.\  639--642.

\bibitem[\protect\citeauthoryear{Dukes and Vansteelandt}{Dukes and
  Vansteelandt}{2021}]{dukes2021inference}
Dukes, O. and S.~Vansteelandt (2021).
\newblock Inference for treatment effect parameters in potentially misspecified
  high-dimensional models.
\newblock {\em Biometrika\/}~{\em 108}, 321--334.

\bibitem[\protect\citeauthoryear{Fang and Shao}{Fang and
  Shao}{2016}]{fang2016model}
Fang, F. and J.~Shao (2016).
\newblock Model selection with nonignorable nonresponse.
\newblock {\em Biometrika\/}~{\em 103}, 861--874.

\bibitem[\protect\citeauthoryear{Filion}{Filion}{1976}]{filion1976exploring}
Filion, F. (1976).
\newblock Exploring and correcting for nonresponse bias using follow-ups of non
  respondents.
\newblock {\em Pacific Sociological Review\/}~{\em 19}, 401--408.

\bibitem[\protect\citeauthoryear{Franks, D’Amour, and Feller}{Franks
  et~al.}{2020}]{franks2020flexible}
Franks, A.~M., A.~D’Amour, and A.~Feller (2020).
\newblock Flexible sensitivity analysis for observational studies without
  observable implications.
\newblock {\em Journal of the American Statistical Association\/}~{\em 115},
  1730--1746.

\bibitem[\protect\citeauthoryear{Garcia, Ibrahim, and Zhu}{Garcia
  et~al.}{2010}]{garcia2010variable}
Garcia, R.~I., J.~G. Ibrahim, and H.~Zhu (2010).
\newblock Variable selection for regression models with missing data.
\newblock {\em Statistica Sinica\/}~{\em 20\/}(1), 149.

\bibitem[\protect\citeauthoryear{Groves and Couper}{Groves and
  Couper}{1998}]{groves1998nonresponse}
Groves, R.~M. and M.~P. Couper (1998).
\newblock {\em Nonresponse in Household Interview Surveys}.
\newblock John Wiley \& Sons.

\bibitem[\protect\citeauthoryear{Guan, Leung, and Qin}{Guan
  et~al.}{2018}]{guan2018semiparametric}
Guan, Z., D.~H. Leung, and J.~Qin (2018).
\newblock Semiparametric maximum likelihood inference for nonignorable
  nonresponse with callbacks.
\newblock {\em Scandinavian Journal of Statistics\/}~{\em 45}, 962--984.

\bibitem[\protect\citeauthoryear{Hansen}{Hansen}{1982}]{hansen1982large}
Hansen, L.~P. (1982).
\newblock Large sample properties of generalized method of moments estimators.
\newblock {\em Econometrica\/}~{\em 50}, 1029--1054.

\bibitem[\protect\citeauthoryear{Heckman}{Heckman}{1979}]{heckman1979sample}
Heckman, J.~J. (1979).
\newblock Sample selection bias as a specification error.
\newblock {\em Econometrica\/}~{\em 47}, 153--161.

\bibitem[\protect\citeauthoryear{Jackson, White, and Leese}{Jackson
  et~al.}{2010}]{jackson2010much}
Jackson, D., I.~R. White, and M.~Leese (2010).
\newblock How much can we learn about missing data?: an exploration of a
  clinical trial in psychiatry.
\newblock {\em Journal of the Royal Statistical Society: Series A\/}~{\em 173},
  593--612.

\bibitem[\protect\citeauthoryear{Kang and Schafer}{Kang and
  Schafer}{2007}]{kang2007demystifying}
Kang, J.~D. and J.~L. Schafer (2007).
\newblock Demystifying double robustness: A comparison of alternative
  strategies for estimating a population mean from incomplete data.
\newblock {\em Statistical Science\/}~{\em 22}, 523--539.

\bibitem[\protect\citeauthoryear{Kennedy, Ma, McHugh, and Small}{Kennedy
  et~al.}{2017}]{kennedy2017non}
Kennedy, E.~H., Z.~Ma, M.~D. McHugh, and D.~S. Small (2017).
\newblock Non-parametric methods for doubly robust estimation of continuous
  treatment effects.
\newblock {\em Journal of the Royal Statistical Society: Series B\/}~{\em 79},
  1229--1245.

\bibitem[\protect\citeauthoryear{Kim and Im}{Kim and
  Im}{2014}]{kim2014propensity}
Kim, J.~K. and J.~Im (2014).
\newblock Propensity score adjustment with several follow-ups.
\newblock {\em Biometrika\/}~{\em 101}, 439--448.

\bibitem[\protect\citeauthoryear{Kim and Yu}{Kim and
  Yu}{2011}]{kim2011semiparametric}
Kim, J.~K. and C.~L. Yu (2011).
\newblock A semiparametric estimation of mean functionals with nonignorable
  missing data.
\newblock {\em Journal of the American Statistical Association\/}~{\em 106},
  157--165.

\bibitem[\protect\citeauthoryear{Kreuter}{Kreuter}{2013}]{kreuter2013improving}
Kreuter, F. (2013).
\newblock {\em Improving surveys with paradata}.
\newblock New Jersey, Hoboken: John Wiley \& Sons.

\bibitem[\protect\citeauthoryear{Lin and Schaeffer}{Lin and
  Schaeffer}{1995}]{lin1995using}
Lin, I.-F. and N.~C. Schaeffer (1995).
\newblock Using survey participants to estimate the impact of nonparticipation.
\newblock {\em Public Opinion Quarterly\/}~{\em 59}, 236--258.

\bibitem[\protect\citeauthoryear{Lipsitz, Ibrahim, and Zhao}{Lipsitz
  et~al.}{1999}]{lipsitz1999weighted}
Lipsitz, S.~R., J.~G. Ibrahim, and L.~P. Zhao (1999).
\newblock A weighted estimating equation for missing covariate data with
  properties similar to maximum likelihood.
\newblock {\em Journal of the American Statistical Association\/}~{\em 94},
  1147--1160.

\bibitem[\protect\citeauthoryear{Liu, Miao, Sun, Robins, and
  Tchetgen~Tchetgen}{Liu et~al.}{2020}]{liu2020identification}
Liu, L., W.~Miao, B.~Sun, J.~Robins, and E.~Tchetgen~Tchetgen (2020).
\newblock Identification and inference for marginal average treatment effect on
  the treated with an instrumental variable.
\newblock {\em Statistica Sinica\/}~{\em 30}, 1517--1541.

\bibitem[\protect\citeauthoryear{Malinsky, Shpitser, and
  Tchetgen~Tchetgen}{Malinsky et~al.}{2020}]{malinsky2020semiparametric}
Malinsky, D., I.~Shpitser, and E.~J. Tchetgen~Tchetgen (2020).
\newblock Semiparametric inference for non-monotone missing-not-at-random data:
  the no self-censoring model.
\newblock {\em Journal of the American Statistical Association\/}.

\bibitem[\protect\citeauthoryear{McFadden}{McFadden}{2001}]{mcfadden2001economic}
McFadden, D. (2001).
\newblock Economic choices.
\newblock {\em American Economic Review\/}~{\em 91}, 351--378.

\bibitem[\protect\citeauthoryear{Miao, Ding, and Geng}{Miao
  et~al.}{2016}]{miao2016identifiability}
Miao, W., P.~Ding, and Z.~Geng (2016).
\newblock Identifiability of normal and normal mixture models with nonignorable
  missing data.
\newblock {\em Journal of the American Statistical Association\/}~{\em 111},
  1673--1683.

\bibitem[\protect\citeauthoryear{Miao and Tchetgen~Tchetgen}{Miao and
  Tchetgen~Tchetgen}{2016}]{miao2016varieties}
Miao, W. and E.~Tchetgen~Tchetgen (2016).
\newblock On varieties of doubly robust estimators under missingness not at
  random with a shadow variable.
\newblock {\em Biometrika\/}~{\em 103}, 475--482.

\bibitem[\protect\citeauthoryear{Mohan and Pearl}{Mohan and
  Pearl}{2021}]{mohan2021graphical}
Mohan, K. and J.~Pearl (2021).
\newblock Graphical models for processing missing data.
\newblock {\em Journal of the American Statistical Association\/}~{\em 116},
  1023--1037.

\bibitem[\protect\citeauthoryear{Molenberghs, Beunckens, Sotto, and
  Kenward}{Molenberghs et~al.}{2008}]{molenberghs2008every}
Molenberghs, G., C.~Beunckens, C.~Sotto, and M.~G. Kenward (2008).
\newblock Every missingness not at random model has a missingness at random
  counterpart with equal fit.
\newblock {\em Journal of the Royal Statistical Society: Series B\/}~{\em 70},
  371--388.

\bibitem[\protect\citeauthoryear{{National Research Council}}{{National
  Research Council}}{2013}]{ces2013measuring}
{National Research Council} (2013).
\newblock Measuring what we spend: Toward a new consumer expenditure survey.
\newblock National Academies Press.

\bibitem[\protect\citeauthoryear{Newey and McFadden}{Newey and
  McFadden}{1994}]{newey1994large}
Newey, W.~K. and D.~McFadden (1994).
\newblock Large sample estimation and hypothesis testing.
\newblock In R.~F. Engle and D.~L. McFadden (Eds.), {\em Handbook of
  Econometrics}, Volume~4, pp.\  2111--2245. Amsterdam: Elsevier.

\bibitem[\protect\citeauthoryear{Ogburn, Rotnitzky, and Robins}{Ogburn
  et~al.}{2015}]{ogburn2015doubly}
Ogburn, E.~L., A.~Rotnitzky, and J.~M. Robins (2015).
\newblock Doubly robust estimation of the local average treatment effect curve.
\newblock {\em Journal of the Royal Statistical Society. Series B\/}~{\em 77},
  373--396.

\bibitem[\protect\citeauthoryear{Okui, Small, Tan, and Robins}{Okui
  et~al.}{2012}]{okui2012doubly}
Okui, R., D.~S. Small, Z.~Tan, and J.~M. Robins (2012).
\newblock Doubly robust instrumental variable regression.
\newblock {\em Statistica Sinica\/}~{\em 22}, 173--205.

\bibitem[\protect\citeauthoryear{Olson}{Olson}{2013}]{olson2013paradata}
Olson, K. (2013).
\newblock Paradata for nonresponse adjustment.
\newblock {\em The Annals of the American Academy of Political and Social
  Science\/}~{\em 645}, 142--170.

\bibitem[\protect\citeauthoryear{Osius}{Osius}{2004}]{osius2004association}
Osius, G. (2004).
\newblock The association between two random elements: A complete
  characterization and odds ratio models.
\newblock {\em Metrika\/}~{\em 60}, 261--277.

\bibitem[\protect\citeauthoryear{Peress}{Peress}{2010}]{peress2010correcting}
Peress, M. (2010).
\newblock Correcting for survey nonresponse using variable response propensity.
\newblock {\em Journal of the American Statistical Association\/}~{\em 105},
  1418--1430.

\bibitem[\protect\citeauthoryear{Politz and Simmons}{Politz and
  Simmons}{1949}]{politz1949attempt}
Politz, A. and W.~Simmons (1949).
\newblock An attempt to get the “not at homes” into the sample without
  callbacks.
\newblock {\em Journal of the American Statistical Association\/}~{\em 44},
  9--16.

\bibitem[\protect\citeauthoryear{Potthoff, Manton, and Woodbury}{Potthoff
  et~al.}{1993}]{potthoff1993correcting}
Potthoff, R.~F., K.~G. Manton, and M.~A. Woodbury (1993).
\newblock Correcting for nonavailability bias in surveys by weighting based on
  number of callbacks.
\newblock {\em Journal of the American Statistical Association\/}~{\em 88},
  1197--1207.

\bibitem[\protect\citeauthoryear{Qin and Follmann}{Qin and
  Follmann}{2014}]{qin2014semiparametric}
Qin, J. and D.~A. Follmann (2014).
\newblock Semiparametric maximum likelihood inference by using failed contact
  attempts to adjust for nonignorable nonresponse.
\newblock {\em Biometrika\/}~{\em 101}, 985--991.

\bibitem[\protect\citeauthoryear{Richardson, Robins, and Wang}{Richardson
  et~al.}{2017}]{richardson2017modeling}
Richardson, T.~S., J.~M. Robins, and L.~Wang (2017).
\newblock On modeling and estimation for the relative risk and risk difference.
\newblock {\em Journal of the American Statistical Association\/}~{\em 112},
  1121--1130.

\bibitem[\protect\citeauthoryear{Robins, Li, Tchetgen~Tchetgen, and van~der
  Vaart}{Robins et~al.}{2008}]{robins2008higher}
Robins, J., L.~Li, E.~Tchetgen~Tchetgen, and A.~van~der Vaart (2008).
\newblock Higher order influence functions and minimax estimation of nonlinear
  functionals.
\newblock In D.~Nolan and T.~Speed (Eds.), {\em Probability and Statistics:
  Essays in Honor of David A. Freedman}, Volume~2, pp.\  335--421. Beachwood,
  Ohio: Institute of Mathematical Statistics.

\bibitem[\protect\citeauthoryear{Robins, Rotnitzky, and Scharfstein}{Robins
  et~al.}{2000}]{robins2000sensitivity}
Robins, J.~M., A.~Rotnitzky, and D.~O. Scharfstein (2000).
\newblock Sensitivity analysis for selection bias and unmeasured confounding in
  missing data and causal inference models.
\newblock In {\em Statistical Models in Epidemiology, the Environment, and
  Clinical Trials}, pp.\  1--94. Springer.

\bibitem[\protect\citeauthoryear{Rotnitzky, Scharfstein, Su, and
  Robins}{Rotnitzky et~al.}{2001}]{rotnitzky2001methods}
Rotnitzky, A., D.~Scharfstein, T.-L. Su, and J.~Robins (2001).
\newblock Methods for conducting sensitivity analysis of trials with
  potentially nonignorable competing causes of censoring.
\newblock {\em Biometrics\/}~{\em 57}, 103--113.

\bibitem[\protect\citeauthoryear{Rotnitzky, Smucler, and Robins}{Rotnitzky
  et~al.}{2021}]{rotnitzky2021characterization}
Rotnitzky, A., E.~Smucler, and J.~M. Robins (2021).
\newblock Characterization of parameters with a mixed bias property.
\newblock {\em Biometrika\/}~{\em 108}, 231--238.

\bibitem[\protect\citeauthoryear{Sadinle and Reiter}{Sadinle and
  Reiter}{2017}]{sadinle2017itemwise}
Sadinle, M. and J.~P. Reiter (2017).
\newblock Itemwise conditionally independent nonresponse modelling for
  incomplete multivariate data.
\newblock {\em Biometrika\/}~{\em 104}, 207--220.

\bibitem[\protect\citeauthoryear{Scharfstein, Rotnitzky, and
  Robins}{Scharfstein et~al.}{1999}]{scharfstein1999adjusting}
Scharfstein, D.~O., A.~Rotnitzky, and J.~M. Robins (1999).
\newblock Adjusting for nonignorable drop-out using semiparametric nonresponse
  models.
\newblock {\em Journal of the American Statistical Association\/}~{\em 94},
  1096--1120.

\bibitem[\protect\citeauthoryear{Stephens, Tchetgen~Tchetgen, and
  De~Gruttola}{Stephens et~al.}{2014}]{stephens2014locally}
Stephens, A., E.~Tchetgen~Tchetgen, and V.~De~Gruttola (2014).
\newblock Locally efficient estimation of marginal treatment effects when
  outcomes are correlated: is the prize worth the chase?
\newblock {\em The International Journal of Biostatistics\/}~{\em 10}, 59--75.

\bibitem[\protect\citeauthoryear{Sun, Liu, Miao, Wirth, Robins, and
  Tchetgen~Tchetgen}{Sun et~al.}{2018}]{sun2018semiparametric}
Sun, B., L.~Liu, W.~Miao, K.~Wirth, J.~Robins, and E.~Tchetgen~Tchetgen (2018).
\newblock Semiparametric estimation with data missing not at random using an
  instrumental variable.
\newblock {\em Statistica Sinica\/}~{\em 28}, 1965--1983.

\bibitem[\protect\citeauthoryear{Tan}{Tan}{2006}]{tan2006distributional}
Tan, Z. (2006).
\newblock A distributional approach for causal inference using propensity
  scores.
\newblock {\em Journal of the American Statistical Association\/}~{\em 101},
  1619--1637.

\bibitem[\protect\citeauthoryear{Tan}{Tan}{2010}]{tan2010bounded}
Tan, Z. (2010).
\newblock Bounded, efficient and doubly robust estimation with inverse
  weighting.
\newblock {\em Biometrika\/}~{\em 97}, 661--682.

\bibitem[\protect\citeauthoryear{Tan}{Tan}{2020}]{tan2020model}
Tan, Z. (2020).
\newblock Model-assisted inference for treatment effects using regularized
  calibrated estimation with high-dimensional data.
\newblock {\em The Annals of Statistics\/}~{\em 48}, 811--837.

\bibitem[\protect\citeauthoryear{Tchetgen~Tchetgen and Wirth}{Tchetgen~Tchetgen
  and Wirth}{2017}]{tchetgen2017general}
Tchetgen~Tchetgen, E.~J. and K.~E. Wirth (2017).
\newblock A general instrumental variable framework for regression analysis
  with outcome missing not at random.
\newblock {\em Biometrics\/}~{\em 73}, 1123--1131.

\bibitem[\protect\citeauthoryear{Train}{Train}{2009}]{train2009discrete}
Train, K.~E. (2009).
\newblock {\em Discrete choice methods with simulation}.
\newblock Cambridge University Press.

\bibitem[\protect\citeauthoryear{Tsiatis}{Tsiatis}{2006}]{tsiatis2006semiparametric}
Tsiatis, A. (2006).
\newblock {\em Semiparametric theory and missing data}.
\newblock New York: Springer.

\bibitem[\protect\citeauthoryear{Tsiatis, Davidian, and Cao}{Tsiatis
  et~al.}{2011}]{tsiatis2011improved}
Tsiatis, A.~A., M.~Davidian, and W.~Cao (2011).
\newblock Improved doubly robust estimation when data are monotonely coarsened,
  with application to longitudinal studies with dropout.
\newblock {\em Biometrics\/}~{\em 67}, 536--545.

\bibitem[\protect\citeauthoryear{van~der Laan and Rubin}{van~der Laan and
  Rubin}{2006}]{vanderlaan2006targeted}
van~der Laan, M.~J. and D.~Rubin (2006).
\newblock Targeted maximum likelihood learning.
\newblock {\em The International Journal of Biostatistics\/}~{\em 2\/}(1).

\bibitem[\protect\citeauthoryear{van~der Vaart and Wellner}{van~der Vaart and
  Wellner}{1996}]{van1996weak}
van~der Vaart, A. and J.~A. Wellner (1996).
\newblock {\em Weak convergence and empirical processes: with applications to
  statistics}.
\newblock Springer Science \& Business Media.

\bibitem[\protect\citeauthoryear{Vansteelandt, Rotnitzky, and
  Robins}{Vansteelandt et~al.}{2007}]{vansteelandt2007estimation}
Vansteelandt, S., A.~Rotnitzky, and J.~Robins (2007).
\newblock Estimation of regression models for the mean of repeated outcomes
  under nonignorable nonmonotone nonresponse.
\newblock {\em Biometrika\/}~{\em 94}, 841--860.

\bibitem[\protect\citeauthoryear{Vermeulen and Vansteelandt}{Vermeulen and
  Vansteelandt}{2015}]{vermeulen2014biased}
Vermeulen, K. and S.~Vansteelandt (2015).
\newblock Biased-reduced doubly robust estimation.
\newblock {\em Journal of the American Statistical Association\/}~{\em 110},
  1024--1036.

\bibitem[\protect\citeauthoryear{Wang, Shao, and Kim}{Wang
  et~al.}{2014}]{wang2014instrumental}
Wang, S., J.~Shao, and J.~K. Kim (2014).
\newblock An instrumental variable approach for identification and estimation
  with nonignorable nonresponse.
\newblock {\em Statistica Sinica\/}~{\em 24}, 1097--1116.

\bibitem[\protect\citeauthoryear{Wood, White, and Hotopf}{Wood
  et~al.}{2006}]{wood2006using}
Wood, A.~M., I.~R. White, and M.~Hotopf (2006).
\newblock Using number of failed contact attempts to adjust for non-ignorable
  non-response.
\newblock {\em Journal of the Royal Statistical Society: Series A\/}~{\em 169},
  525--542.

\bibitem[\protect\citeauthoryear{Zhang, Chen, and Zhang}{Zhang
  et~al.}{2018}]{zhang2018bayesian}
Zhang, Y., H.~Chen, and N.~Zhang (2018).
\newblock Bayesian inference for nonresponse two-phase sampling.
\newblock {\em Statistica Sinica\/}~{\em 28}, 2167--2187.

\end{thebibliography}


\begin{thebibliography}{}

\bibitem[\protect\citeauthoryear{Kim and Im}{Kim and
  Im}{2014}]{kim2014propensity}
Kim, J.~K. and J.~Im (2014).
\newblock Propensity score adjustment with several follow-ups.
\newblock {\em Biometrika\/}~{\em 101}, 439--448.

\bibitem[\protect\citeauthoryear{Newey and McFadden}{Newey and
  McFadden}{1994}]{newey1994large}
Newey, W.~K. and D.~McFadden (1994).
\newblock Large sample estimation and hypothesis testing.
\newblock In R.~F. Engle and D.~L. McFadden (Eds.), {\em Handbook of
  Econometrics}, Volume~4, pp.\  2111--2245. Amsterdam: Elsevier.

\bibitem[\protect\citeauthoryear{Tsiatis}{Tsiatis}{2006}]{tsiatis2006semiparametric}
Tsiatis, A. (2006).
\newblock {\em Semiparametric theory and missing data}.
\newblock New York: Springer.

\bibitem[\protect\citeauthoryear{van~der Vaart}{van~der
  Vaart}{2000}]{van2000asymptotic}
van~der Vaart, A.~W. (2000).
\newblock {\em Asymptotic Statistics}.
\newblock Cambridge University Press.

\bibitem[\protect\citeauthoryear{Vansteelandt, Rotnitzky, and
  Robins}{Vansteelandt et~al.}{2007}]{vansteelandt2007estimation}
Vansteelandt, S., A.~Rotnitzky, and J.~Robins (2007).
\newblock Estimation of regression models for the mean of repeated outcomes
  under nonignorable nonmonotone nonresponse.
\newblock {\em Biometrika\/}~{\em 94}, 841--860.

\end{thebibliography}

%
%
%

%
%
%

\end{document}